%% file: main.tex
\documentclass[11pt]{article}

\usepackage{microtype}
\usepackage{graphicx}
\usepackage{booktabs} % for professional tables

\usepackage{tabularx}

\usepackage{natbib}
\usepackage{fullpage}
\usepackage[utf8]{inputenc}
\usepackage[T1]{fontenc}
\usepackage{microtype}
\usepackage{csquotes}
\usepackage[osf,sc]{mathpazo}
\RequirePackage[scaled=0.90]{helvet}
\RequirePackage[scaled=0.85]{beramono}
\RequirePackage{textcomp}
\usepackage{xcolor}
\definecolor{DarkGreen}{rgb}{0.075,0.375,0.075}
\definecolor{DarkRed}{rgb}{0.5,0.1,0.1}
\definecolor{DarkBlue}{rgb}{0.1,0.1,0.5}
\definecolor{Gray}{rgb}{0.2,0.2,0.2}
\usepackage[pdftex]{hyperref}
\hypersetup{
    unicode=false,          % non-Latin characters in AcrobatÂ¿s bookmarks
    pdftoolbar=true,        % show Acrobat toolbar?
    pdfmenubar=true,        % show Acrobat menu?
    pdffitwindow=false,      % page fit to window when opened
    pdfnewwindow=true,      % links in new window
    colorlinks=true,       % false: boxed links; true: colored links
    linkcolor=DarkBlue,          % color of internal links
    citecolor=DarkGreen,        % color of links to bibliography
    filecolor=DarkRed,      % color of file links
    urlcolor=DarkBlue,          % color of external links
    pdftitle={},
    pdfauthor={},
}
\usepackage[skip=2mm,indent=5mm]{parskip}
\usepackage[small]{caption}
\RequirePackage{amsthm,amsmath,amsfonts,amssymb}

\usepackage[capitalize,noabbrev]{cleveref}

\usepackage[utf8]{inputenc} % allow utf-8 input
\usepackage[T1]{fontenc}    % use 8-bit T1 fonts
\usepackage{hyperref}       % hyperlinks
\usepackage{url}            % simple URL typesetting
\usepackage{booktabs}       % professional-quality tables
\usepackage{amsfonts}       % blackboard math symbols
\usepackage{nicefrac}       % compact symbols for 1/2, etc.
\usepackage{microtype}      % microtypography
\usepackage{xcolor}         % colors
\usepackage{amsmath}
\usepackage{amsthm} %Theorems
\usepackage{enumitem}
\usepackage{pifont}
\usepackage{wrapfig}
\usepackage{bm}
\usepackage{bbm}
\usepackage{todonotes}
\usepackage{makecell}

\usepackage{caption}
\usepackage{subcaption}

\input{math_commands}

\usepackage[ruled,vlined]{algorithm2e}
\usepackage{etoolbox}

\theoremstyle{plain}
\newtheorem{theorem}{Theorem}[section]
\newtheorem{proposition}[theorem]{Proposition}
\newtheorem{lemma}[theorem]{Lemma}
\newtheorem{corollary}[theorem]{Corollary}
\theoremstyle{definition}
\newtheorem{definition}[theorem]{Definition}
\newtheorem{assumption}[theorem]{Assumption}
\theoremstyle{remark}

\newtheorem{example}[theorem]{Example}

\newcommand{\etbt}{\Delta^{\textit{tbt}}}

\newcommand{\etbtrstar}{\Delta_r^{\textit{tbt}^\star}}
\newcommand{\rank}{rank}
\newcommand{\req}{{\text{req}}}

\usepackage[normalem]{ulem}

\title{Leaderboard Incentives:\\ Model Rankings under Strategic Post-Training}

\author{%
  Yatong Chen\footnote{Corresponding author: \href{mailto:yatong.chen@tuebingen.mpg.de}{yatong.chen@tuebingen.mpg.de}}
  , Guanhua Zhang, Moritz Hardt}
\date{
\textit{Max Planck Institute for Intelligent Systems, Tübingen} and \textit{Tübingen AI Center}\\
}

\newcounter{daggerfootnote}

\begin{document}

\onecolumn
\maketitle

\begin{abstract}
Influential benchmarks incentivize competing model developers to strategically allocate post-training resources towards improvements on the leaderboard, a phenomenon dubbed \emph{benchmaxxing} or \emph{training on the test task}. In this work, we initiate a principled study of the incentive structure that benchmarks induce. We model benchmarking as a Stackelberg game between a benchmark designer who chooses an evaluation protocol and multiple model developers who compete simultaneously in a subgame given by the designer’s choice. Each competitor has a model of unknown latent quality and can inflate its observed score by allocating resources to benchmark-specific improvements. First, we prove that current benchmarks induce games for which no Nash equilibrium between model developers exists. This result suggests one explanation for why current practice leads to misaligned incentives, prompting model developers to strategize in opaque ways. However, we prove that under mild conditions, a recently proposed evaluation protocol, called tune-before-test, induces a benchmark with a unique Nash equilibrium that ranks models by latent quality. This positive result demonstrates that benchmarks need not set bad incentives, even if current evaluations do.
\end{abstract}

\input{Sections/introv2}
\input{Sections/related-work}

\input{Sections/model}

\input{Sections/equilibrium}

\input{Sections/tune-before-test}

\input{Sections/conclusion}

\input{Sections/acknowledgement}

\bibliographystyle{plainnat}
\bibliography{references}
\newpage
\appendix
\input{Sections/appendix}

\end{document}

%% file: math_commands.tex
\def\eqref#1{equation~\ref{#1}}
\def\1{\bm{1}}

\DeclareMathAlphabet{\mathsfit}{\encodingdefault}{\sfdefault}{m}{sl}
\SetMathAlphabet{\mathsfit}{bold}{\encodingdefault}{\sfdefault}{bx}{n}

%% file: Sections/introv2.tex
\section{Introduction}

Traditionally, machine learning benchmarks came with a fixed training set, requiring that all models under comparison train on the same data. The situation has changed with large language model benchmarks that typically only provide test data, leaving the choice of training data to the model developer. This has raised the concern that model developers can inflate benchmark performance with benchmark-specific tweaks that don’t broadly improve model capabilities. The resulting problem, called benchmaxxing or training on the test task, confounds model comparisons and may cause misleading leaderboards, as prior work shows \citep{dominguez2024training, singh2025leaderboard}. But the situation need not be a sign of cheating or wrongdoing. Rather, influential benchmarks incentivize competing model developers to strategically allocate post-training resources towards improvements on the leaderboard. Although widely recognized, there is currently no formal understanding of the incentives that benchmarks set.

In this work, we initiate a principled study of the incentive structure that benchmarks induce. We model benchmarking as a Stackelberg game between a benchmark designer and multiple competing model developers. The designer chooses an evaluation protocol, and the model developers compete in a simultaneous-move subgame given by the designer’s choice. Each competitor has a model of a latent quality --- unknown to the benchmark designer --- and can inflate the observed model score by allocating additional resources to benchmark-specific improvements at a cost. The designer aims to choose an evaluation protocol so that the resulting benchmark yields a ranking by latent quality at equilibrium when competitors best respond to each other.

\subsection{Our Results}

Our first result is descriptive and negative: Current benchmarks induce games in which generally no Nash equilibrium between model developers exists. Model developers are always incentivized to strategize in opaque ways, leading to uninterpretable leaderboards that may not reflect a ranking by latent quality. In contrast, our second result is prescriptive and positive. Under mild conditions, the designer has a cost-effective evaluation strategy that induces a subgame with a unique Nash equilibrium that ranks models by latent quality. In addition, at this equilibrium solution, model developers refrain from benchmaxxing altogether, investing no additional effort in benchmark-specific improvements.

What makes the incentive design problem challenging is that the benchmark creator has no control over the utility each model developer has for gains on the leaderboard, does not know latent capabilities, and cannot limit how much effort model developers invest. The evaluation protocol we study tunes each model on the same small amount of task-specific data before evaluation. The intuition is that a small amount of task-specific preparation levels the playing field by washing out minor benchmark-specific tweaks. Previous work showed empirically that this intervention, called tune-before-test (TbT), leads to consistent model rankings across a wide range of benchmarks \citep{zhang2025train}. Our main result adds a game-theoretic justification: Tune-before-test creates incentives that lead to rankings by latent quality at equilibrium. Surprisingly, a small amount of data suffices to realize the same effect as orders of magnitude more training, as we show.

Complementing our theoretical results, we demonstrate empirically that the assumptions of our theorem hold in a case study of a representative benchmark. After applying tune-before-test with only 3,000 steps, a model developer would have to invest at least 384,668 additional training steps to change model rankings. {This captures why TbT can be such an effective intervention: by pushing all models into a diminishing-returns regime, it greatly increases the marginal cost of further score improvements. As a result, overtaking nearby competitors requires substantially more additional effort, amplifying the asymmetry in local overtaking incentives predicted by our theory (see \Cref{fig:traj} right).}

%% file: Sections/related-work.tex
\subsection{Related Work}
Benchmarks have been the key driver of machine learning progress by enabling frictionless comparison and competition~\citep{Donoho2023DataSA} between models. 
In the traditional supervised learning paradigm, benchmarks typically come with a fixed training set and a held-out test set~\citep{Garofolo1993DarpaTA,lecun-mnisthandwrittendigit-2010,TjongKimSang2003IntroductionTT,liberman2015reproducible,hardtrecht2022patterns}, enabling relatively controlled comparisons across methods.
Although test-set reuse~\citep{Duda1974PatternCA} can erode classical statistical guarantees~\citep{Dwork2015GeneralizationIA,Dwork2015TheRH,Blum2015TheLA,Mania2019ModelSM}, fixed train/test splits have historically supported relatively robust model comparison and have helped establish widely adopted baselines, exemplified by the ImageNet challenge~\citep{Deng2009ImageNetAL,Russakovsky2014ImageNetLS} and its role in accelerating the adoption of deep learning~\citep{Krizhevsky2012ImageNetCW,Goodfellow-et-al-2016,He2015DeepRL}.

A key reason benchmarks are useful in practice is \emph{external validity} \citep{Liao2021AreWL,Salaudeen2025MeasurementTM}: performance and rankings on one benchmark often correlate with performance on related datasets and tasks~\citep{Yadav2019ColdCT,Recht2019DoIC,Miller2020TheEO}, allowing practitioners to select strong models with some confidence. 
For instance, studies have shown that ImageNet rankings transfer well to other image datasets~\citep{Kornblith2018DoBI,Salaudeen2024ImageNotAC}.
In the large language model (LLM) era, however, many influential benchmarks provide primarily test instances and a scoring protocol~\citep{Patwardhan2025GDPvalEA,Glazer2024FrontierMathAB,Zhou2023InstructionFollowingEF,Jain2024LiveCodeBenchHA}, leaving training data and post-training choices largely unconstrained~\citep{Raffel2019ExploringTL,Albalak2024ASO,Guha2025OpenThoughtsDR,Li2024DataCompLMIS}. 
This change expands the space of strategic choices available to model providers and can lead to substantial ranking variation across benchmarks~\citep{Huan2025DoesMR,Zhang2024InherentTB,Liang2023HolisticEO,open-llm-leaderboard-v2, hardt2025emerging}, even among benchmarks that aim to measure similar capabilities.

One mechanism behind the ranking variation is what \citet{dominguez2024training} call \emph{training on the test task}. 
Model providers conduct benchmark-aware post-training~\citep{Touvron2023Llama2O} and data curation~\citep{Guha2025OpenThoughtsDR} to achieve large gains on measured tasks without comparable improvements in general capability~\citep{Zhou2023DontMY,singh2025leaderboard}. 
Training on the test task differs from training on the test set~\citep{Yang2023RethinkingBA}, or data contamination~\citep{Jiang2024InvestigatingDC,Yang2023RethinkingBA,Bordt2024HowMC}, where models train directly on the test data. 
Instead, it captures optimization to the evaluation’s task distribution or protocol, such as curating instruction data that matches a benchmark’s format and rubrics or tuning with feedback aligned to its scoring procedure~\citep{dominguez2024training}.
Because such task-level alignment is difficult to detect and is often not explicitly ruled out by benchmark rules, it can confound leaderboard interpretation: scores conflate latent capability~\citep{ruan2024observational} with benchmark-specific effort~\citep{Schaeffer2023AreEA}. 
Model providers can strategically choose how much benchmark-specific effort to invest per benchmark, which yields inconsistent rankings. 

A natural way to mitigate benchmark-specific advantages is to reduce heterogeneity in how models are prepared for the evaluation. 
\citet{zhang2025train} propose tune-before-test (TbT), in which all models are fine-tuned on benchmark-specific data before evaluation to equalize preparation. 
Their empirical results show that TbT can restore ranking consistency across tasks from different domains and reveal that post-TbT scores are dominated by a low-dimensional latent capability factor. 
While \citet{zhang2025train} motivate TbT as a post-hoc correction for observed rankings, we study it as an ex ante mechanism design choice~\citep{Manheim2018CategorizingVO}. In our setting, a benchmark designer commits to an evaluation protocol (including a TbT baseline), anticipating that competing model providers will respond strategically by allocating additional post-training effort.
We ask when an equilibrium exists in this competition and whether equilibrium rankings recover a ranking by latent capability. Treating TbT as a designer-controlled policy parameter instead allows us to isolate its incentive effects and ask how much evaluator-side adaptation is required to stabilize rankings.

Relevant to our work is the literature on strategic classification \citep{bruckner2011stackelberg, hardt2016strategic}, which studies decision-making when individuals may adapt their features in response to a deployed model. See \citet{rosenfeld2024strategicml} for a survey. These interactions are typically modeled as a Stackelberg game, where a decision-maker commits to a classifier and strategic individuals best respond myopically to the classifier. An important question in this literature is how to incentivize genuine improvements rather than superficial \emph{gaming} \citep{kleinberg2020classifiers, miller2020strategic, tal2020multiagent, chen2023learning}. Our work shares this focus on incentives. In a departure from the classification setting, however, ranking inherently creates competition between participants: when one model moves up in rank, another must move down. The leaderboard therefore induces a game between competitors. This aspect of our work also connects to recent work on markets induced by predictive systems~\cite{einav2025market, sommer2025learning}. Closely related is the work on \emph{strategic ranking} \cite{liu2022strategic}, where a continuous population of individuals competes for acceptance by an institution, such as a school or an employer. The institution may set the rewards for achieving a given quantile rank within the continuous population. The results characterize how quantile-based incentive design shapes effort, welfare, and fairness at equilibrium. Our work shares this focus on rank-based competition, but studies the setting of machine learning benchmarking. Here, the designer cannot directly control the downstream rewards associated with a higher rank, such as more attention or increased reputation, since these are exogenous. Instead, the designer’s primary lever is not the reward design but the evaluation protocol itself, which needs to make sense for a machine learning benchmark. Second, \citet{liu2022strategic} model a continuum of agents, leading to a smooth equilibrium characterization with monotone effort profiles. In contrast, leaderboard competition typically involves a finite number of strategic model providers, where small performance changes can swap adjacent ranks and generate discontinuous payoff changes. These local ``just-overtake'' incentives can, in our setting, make equilibrium stability itself a central design concern. Finally, strategic ranking also relates to a vast literature in economics on designing tournaments and contests, see~\cite{corchon2007theory,connelly2014tournament, fu2019contests}, that we cannot survey here.

%% file: Sections/model.tex
\section{Preliminaries}
\paragraph{Capability, effort, and post-effort score.}
Consider a set of LLM model developers $N=\{1,\dots,n\}$ who submit a model to a benchmark leaderboard curated by a leaderboard designer. Each model has a latent variable~$\theta_i\in\mathbb{R}_{\geq 0}$ capturing general model capabilities. This abstraction is empirically motivated by evidence that performance across diverse language-model benchmarks admits a low-dimensional structure and is often dominated by a general capability factor \citep{ruan2024observational, zhang2025train}. The setting extends straightforwardly to a multi-dimensional latent capabilities vector $\vec{\theta}_i\in\mathbb{R}_{\geq 0}^d$ with a benchmark-specific coefficient vector $\mathbf{w}$ so that $\theta_i=\langle \mathbf{w}, \vec{\theta}_i\rangle.$ Without loss of generality, we index models in decreasing order of latent capability, i.e., $\theta_1 > \theta_2 > \cdots > \theta_n$. Capabilities are known to model developers but not to leaderboard designer.

Before evaluation, a model may undergo additional \emph{benchmark-specific adaptation}. %, i.e., training procedures that exploit knowledge of the benchmark task beyond the base pre-training objective \cite{dominguez2024training}.
We represent the \emph{total} amount of benchmark-specific preparation by a scalar $e\in \mathbb{R}_{\geq 0}$, which may be performed by the model developer, and/or applied uniformly by the leaderboard designer (as in tune-before-test).
Following the empirical tune-before-test approach of \cite{zhang2025train}, we use the amount of benchmark-specific
fine-tuning data (e.g., number of training examples) as a concrete and measurable proxy for this adaptation. We model benchmark performance through a \emph{post-effort score} mapping $v=v(\theta,e)$:

\begin{definition}[Post-Effort Score]
\label{def:post-effort-score}
The \emph{post-effort score} of a model with capability $\theta$ and total benchmark-specific training effort $e$ is
$
v(\theta,e): \mathbb{R}_{\geq 0}\times \mathbb{R}_{\geq 0}\to [0,1]$.

\end{definition}

The function $v(\theta,e)$ captures how capability and benchmark-specific effort translate into its actual benchmark performance.
For simplicity, we model $v$ deterministically; equivalently, our results can be interpreted in terms of expected scores under evaluation noise.

\paragraph{Leaderboard rewards.}
Let $\{R_j\}_{j=1}^n$ denote the reward assigned to the model ranked $j$, where rewards are non-increasing in rank, i.e., 
$R_1 \ge R_2 \ge \cdots \ge R_n \ge 0$.  This means higher-ranked models receive greater downstream benefits (e.g., more downstream users, higher visibility) from their leaderboard positions.

\begin{example}
\label{example:leaderboard-reward}
Two canonical leaderboard reward settings:
\begin{enumerate}
    \item \emph{Winner-take-all}: only the top-ranked model receives a reward, i.e., $R_1 > 0$ and $R_j = 0$ for all $j \neq 1$.
    \item \emph{Top-$k$ rewards}: the top $k$ models receive the same reward, i.e., $R_j = R > 0$ for $j \in [k]$ and $R_j = 0$ for $j > k$.
\end{enumerate}
\end{example}

\paragraph{Tune-before-test methodology.}
We model tune-before-test (TbT) as a designer-chosen baseline effort $\etbt \ge 0$ that applies the same amount of benchmark-specific fine-tuning to every submitted model prior to evaluation. We treat TbT as an explicit component of the evaluation protocol: the leaderboard designer commits to $\etbt$, and model developers choose any additional benchmark-specific effort in response.
Setting $\etbt = 0$ recovers standard evaluation without TbT, so the conventional leaderboard protocol is a special case of our framework.
 
\section{Stackelberg Ranking Game}
\label{sec:game}
We model leaderboard evaluation as a Stackelberg game with a single leader and multiple followers \citep{von2010leadership}. The leader (the leaderboard designer) first commits to an evaluation protocol, after which the followers (model developers) simultaneously choose benchmark-specific effort. We study outcomes in which the developers play a Nash equilibrium of the induced follower game, and the leader chooses its action, anticipating this equilibrium response. Such an outcome is often referred to as a \emph{Stackelberg–Nash equilibrium} \citep{marchesi2021leadership}.

\begin{definition}[Stackelberg Ranking Game]
The players are a leaderboard designer (the leader) and a set of model developers $N=\{1,\dots,n\}$ (the followers). Model developers have latent capabilities $(\theta_i)_{i=1}^n$, which are common knowledge among the model developers but unknown to the designer. A public rank-based reward $\{R_j\}_{j=1}^n$ assigns reward $R_j$ to the model ranked $j$.
The sequence of actions is as follows:
\begin{itemize}
\item The designer commits to a benchmark and a tune-before-test baseline $\etbt \ge 0$.
\item After observing $\etbt$, developers simultaneously choose additional effort levels $e_i \ge 0$.
\end{itemize}

Scores are realized as $v_i = v(\theta_i, e_i + \etbt)$, and models are ranked by $\textbf{v}=(v_1,\dots,v_n)$ (ties broken deterministically). Both parties' utilities are specified in \Cref{def:agent-utility} and \Cref{def:leader-utility}.

\end{definition}

\subsection{Model Developer's Utility and Externalities}
\label{sec:agent-utility}
Model developers are modeled as rational agents whose utility equals the rank-based reward minus the cost of additional effort. 
We use a cost function $c:\mathbb{R}_{\geq 0}\to\mathbb{R}_{\geq 0}$ to capture the effort costs \footnote{We assume a common cost function for all model developers for simplicity. In \Cref{prop:equivalence-cost-fn}, we show that \emph{multiplicatively separable} costs ($C_i(e_i;\theta_i)=\gamma_i\,c(e_i)$) yield an equivalent follower game after rescaling rewards.}. Because rewards depend on rank, incentives are inherently interdependent: improving one model’s rank necessarily lowers another’s. As a result, a developer’s optimal effort depends not only on its own capability and cost, but also on competitors’ efforts. Let $\textbf{e}_{-i}$ denote the vector of the efforts of all models except $i$.
%Let $v_{-i}$ denote the vector of post-effort scores of all models except $i$. 
We define the utility of developer $i$ as follows:

\begin{definition}[Model Developer's Utility]
\label{def:agent-utility}
Given a rank-based reward $(R_j)_{j=1}^n$ and a fixed tune-before-test adjustment level $\etbt\ge 0$, the utility of model $i$ with capability $\theta_i$ and effort $e_i$ is
\begin{align}
U_i(e_i; \textbf{e}_{-i}, \etbt)
=
%\lambda\!\left(\operatorname{rank}(v_i)\right)
R_{\operatorname{rank}(v_i)}
- c(e_i),
\end{align}
%where
\begin{align*}
    v_i = v(\theta_i,\etbt+e_i), ~
\operatorname{rank}(v_i)
:= 1 + \sum_{j\in N\setminus\{i\}}\mathbbm{I}\!\left\{v_j>v_i\right\}.
% ,
% \qquad
% \lambda(j):=R_j. 
\end{align*}
\end{definition}

\paragraph{Nash equilibrium of the induced follower game.}

For a fixed TbT level $\etbt$, the developers’ interaction forms the induced follower game, for which we use pure-strategy Nash equilibrium \citep{nash1950equilibirum} as the solution concept.
Intuitively, an effort profile is at equilibrium if no developer can improve its utility by unilaterally changing the amount of benchmark-specific effort it invests:

\begin{definition}[Follower Game's Nash Equilibrium]
    Fix a TbT level $\etbt$. An effort profile $\textbf{e}^*=(e_1^*,\cdots,e_n^*)$ is a \emph{pure-strategy Nash equilibrium} (PNE) of the induced follower game if, for every model developer $i \in N$, 
    $$
    e_i^* \in \arg\max_{e\ge 0} U_i(e; \textbf{e}_{-i}^*, \etbt),
    $$
where $\textbf{e}_{-i}^*$ denotes the equilibrium efforts of other models.
\end{definition}

Although real-world leaderboard competition involves repeated submissions over time, equilibrium analysis still provides a useful baseline: it clarifies whether incentives can ever settle, or whether the leaderboard instead induces persistent “arms-race” behavior.

\subsection{Leaderboard Designer's Utility}
While model developers seek to maximize rank-based rewards, the leaderboard designer evaluates a benchmark by how well its induced ranking reflects the models’ latent capabilities. 
Because the designer only observes benchmark scores, it aims to choose an evaluation protocol that yields a ranking aligned with the capability ordering, while incurring a cost for interventions such as tune-before-test. A minimal way to capture this objective is through a \emph{ranking correctness} criterion:

\begin{definition} (Leaderboard Designer's Utility)
\label{def:leader-utility}
Given a tune-before-test level $\etbt$ and a follower's effort profile $\textbf{e} = (e_1, \cdots, e_n)$ inducing post-effort scores $v_i = v(\theta_i, e_i + \etbt)$, the leaderboard designer's utility is:
\begin{align}
\label{eqn:leader-utility}
    U^L(\etbt; \textbf{e}) = R^L \cdot \mathbbm{I}[\rank(v_i) = \rank(\theta_i), \forall i] - n\cdot c^L(\etbt),
\end{align}
where $R^L \gg 0$ is the reward from achieving a capability-consistent ranking, and $c^L(\etbt)$ is the per-model cost of applying tune-before-test level $\etbt$.
\end{definition}

This binary formulation captures the designer’s core concern: whether the leaderboard correctly orders models by latent capability. One could alternatively model the designer’s objective more smoothly using a rank-correlation metric such as Kendall’s $\tau$ \citep{kendall1938rank}, which rewards partial agreement between score-based and capability rankings, minus the cost of tune-before-test adjustment.

We assume that achieving a capability-consistent ranking is the designer’s primary objective, in the sense that $R^L$ is large enough that any capability-consistent ranking is preferred to any inconsistent one, even after accounting for tune-before-test costs. We also assume that $c^L(\etbt)$ is increasing in $\etbt$. Under these preferences, among all tune-before-test levels that induce a follower equilibrium preserving the capability ordering, the designer prefers the cheapest one. The leader’s optimal choice then corresponds to the \emph{Stackelberg--Nash equilibrium} defined below. 

\subsection{Stackelberg-Nash Equilibrium}
Formally, a pair $({\etbt}^\star, \textbf{
e}^\star)$ is a \emph{Stackelberg--Nash equilibrium} if: (i) $\textbf{e}^\star$ is a pure-strategy Nash equilibrium of the induced follower game under ${\etbt}^\star$, and (ii) ${\etbt}^\star$ maximizes the designer’s utility anticipating this response.

We begin by fixing $\etbt$ and analyzing the induced follower game. Understanding the existence of follower equilibria and the rankings they generate is a prerequisite for determining the designer’s optimal policy. In \Cref{sec:tbt}, we then study how to choose the optimal TbT level ${\etbt}^\star$
 to maximize the leaderboard designer’s utility.

%% file: Sections/equilibrium.tex
\section{Equilibrium Analysis of the Follower Game}
\label{sec:equilibrium}
What incentives does a rank-based leaderboard create for the strategic model developers? 
When model developers can invest in additional benchmark-specific post-training to improve performance, a central question is whether the induced competition admits a pure-strategy Nash equilibrium, that is, whether incentives can settle at a stable effort profile.

Perhaps surprisingly, our first result is partially optimistic: whenever a pure-strategy equilibrium of the induced follower subgame exists, the resulting leaderboard ranking must preserve the latent capability ordering (\Cref{prop:order-preserve}). In this case, any strategic post-training should still lead to a stable and correct ranking.
However, as we point out later, the more subtle issue is that such equilibria need not exist. We show that equilibrium existence depends critically on the \emph{reward gaps} between adjacent ranks (\Cref{thm:no-equilibrium}). When rewards are sufficiently flat, developers face persistent incentives to “just overtake” nearby competitors, leading to arms-race dynamics. %rather than settling down.

\subsection{Assumptions and Empirical Verifications}
To analyze equilibrium behavior in the induced follower game, we need to make some structural assumptions on (i) the cost of benchmark-specific post-training, and (ii) how benchmark performance depends on a model’s capability and post-training effort.

Initial benchmark-specific data curation may be inexpensive, but getting more data typically becomes progressively more costly 
\citep{datologyai2024textcuration}. To capture this, we assume that the cost of effort is non-decreasing and convex: 
\begin{assumption}[Cost Function $c$]
\label{ass:cost}
The cost of benchmark-specific post-training $c:\mathbb{R}_{\geq 0}\to\mathbb{R}_{\geq 0}$ is non-decreasing and convex, with $c(0)=0$ and $\lim_{e\to\infty}c(e)=\infty$.
\end{assumption}

% \paragraph{Post-effort score function.}
Next, we impose structure on the post-effort score function: %$v(\theta,e)$:
\begin{figure*}
    \centering
    \includegraphics[width=0.45\linewidth]{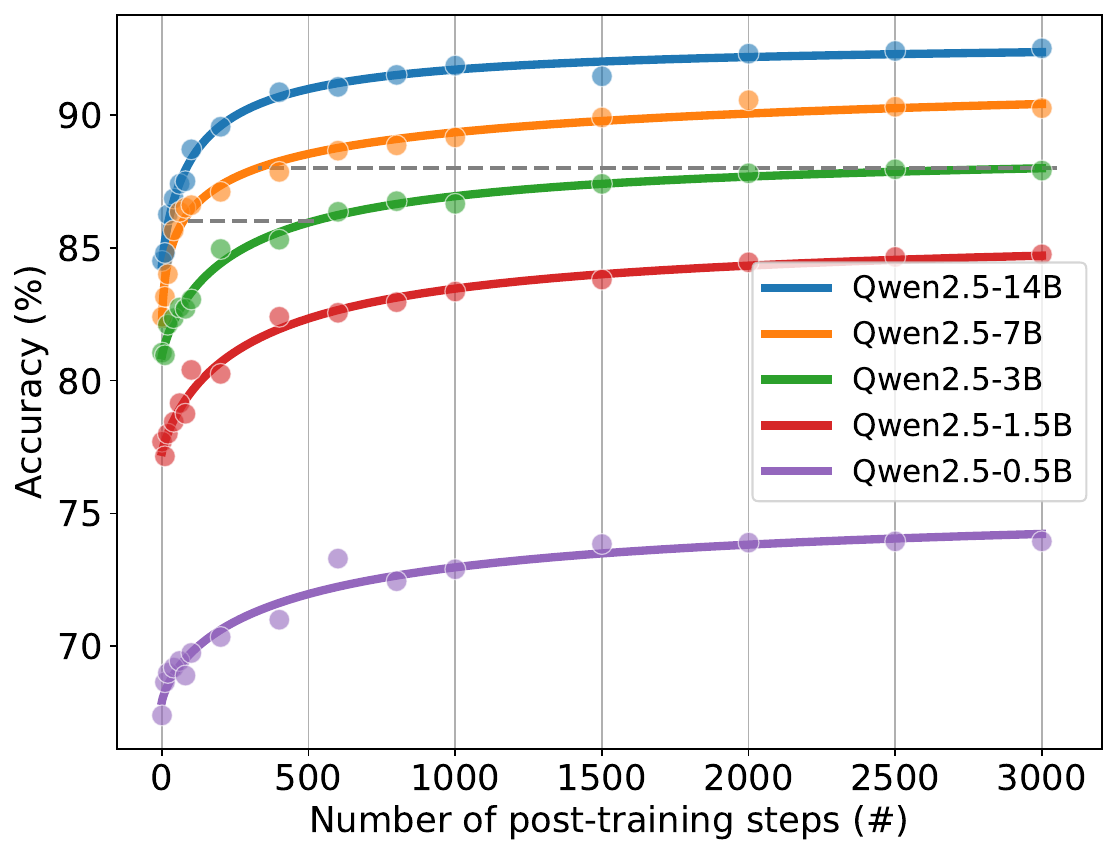}
    \includegraphics[width=0.45\linewidth]{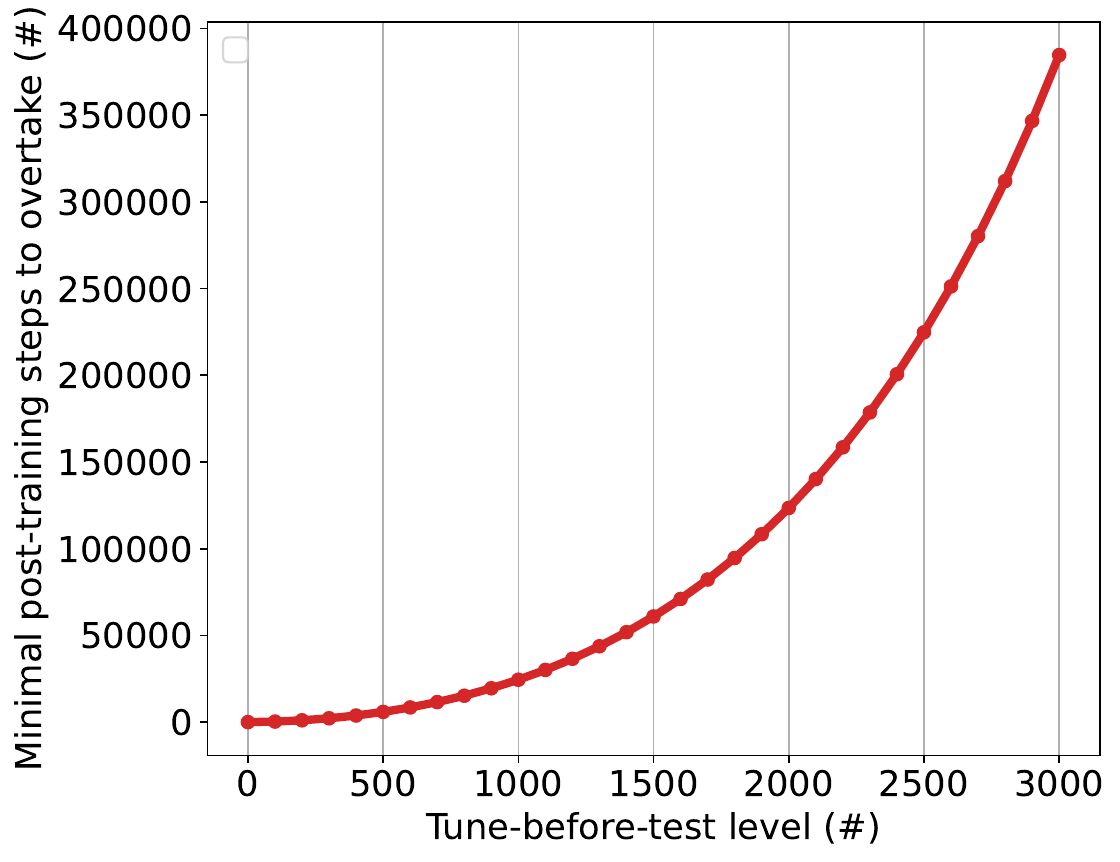}
    \caption{
    \textbf{Left:}
    Continued post-training trajectories of \textit{Qwen2.5} models of different sizes on \texttt{Winogrande}. Here, we use model size as a proxy for the model’s latent capability $\theta$. 
    The $x$-axis denotes the amount of post-training steps (each step corresponds to 8 data points), reflecting post-training effort $e$. 
    The $y$-axis denotes accuracy on the validation set, i.e., $v(\theta, e)$. 
    For each model, we fit a curve following \Cref{eq:log-score}. 
    The empirical results align with the assumptions of \emph{monotonicity in capability}, \emph{diminishing returns and saturation in effort}, and \emph{non-decreasing effort gaps} in Assumption~\ref{ass:post-effort-score}.
    See Appendix~\ref{app:empirical} for additional details and results for the other eight benchmarks. 
    \textbf{Right:} For each tune-before-test level $\etbt$ (the amount of benchmark-specific finetuning steps, $x$-axis), we calculate the minimal additional steps required ($y$-axis) to change the ranking for at least one model, i.e., $\min_{r\in \{2,\ldots,n\}} {e}^\req_r(\etbt)$, based on the fitted curves on the left.
    With $\etbt=3,000$, at least 384,668 training steps are needed to change the ranking of one model.
    }
    \label{fig:traj}
\end{figure*}

\begin{assumption}[Post-effort score function $v$]
\label{ass:post-effort-score}
Let $v:\Theta\times\mathbb{R}_{\geq 0}\to[0,1]$ denote the post-effort score function, where
$\theta\in\Theta\subseteq\mathbb{R}_{\geq 0}$ is the capability and $e\in\mathbb{R}_{\geq 0}$ is the total effort.
Assume $v$ is continuous and that for all $(\theta,e)$:
\begin{enumerate}
    \item[C1.] ({Monotonicity in capability})
    Holding effort fixed, higher capability yields a higher score:
    $\partial_\theta v(\theta,e) > 0$.
    \item[C2.] ({Diminishing returns and saturation in effort})
    Effort weakly improves performance but with diminishing marginal returns, and scores converge to a finite limit:
    $\partial_e v(\theta,e)\ge 0$, $\partial_{ee} v(\theta,e)\le 0$, and
    $v^\infty(\theta):=\lim_{e\to\infty} v(\theta,e)$ exists and is finite.
    \item[C3.] (Non-decreasing effort gaps)
    Let $e^\req(s;\theta):=\inf\{e\geq 0:\ v(\theta,e)\ge s\}$ denote the minimal effort required for capability $\theta$
    to reach target score $s\in[0,1]$.
    For any $\theta'>\theta$, the effort advantage of higher capability, $ e^\req(s;\theta)- e^\req(s;\theta')$, is (weakly) nondecreasing in $s$.
\end{enumerate}
\end{assumption}

These conditions capture three widely observed regularities in post-training scaling behavior: higher-capability models perform better at any fixed effort (C1), benchmark-specific training exhibits diminishing returns and saturation as models approach the benchmark’s performance ceiling (C2).
C3 further requires that as the target score increases, it does not become easier for a lower-capability model to close the gap through post-training alone. 
This is a standard single-crossing–type regularity condition in economic theory \citep{topkis1998supermodularity}, and is consistent with empirical observations that stronger foundation models tend to make more effective use of additional training, especially at a higher target score \citep{wei2021finetuned}. 

\paragraph{Example: generalized power-law scaling.}

As an illustrative functional form consistent with Assumption~\ref{ass:post-effort-score}, consider a generalized scaling law motivated by empirical studies of post-training behavior \citep{ruan2024observational,finnveden2020extrapolating,owen2024predictable}. Let
\[
\tilde v(\theta,e)
:= \frac{v(\theta,e)-L(\theta)}{U(\theta)-L(\theta)} \in [0,1]
\]
denote a normalized score, where $L(\theta)$ and $U(\theta)$ represent model-specific lower and upper performance levels. 
Suppose
\begin{align}
\label{eq:log-score}
    \sigma^{-1}\!\bigl(\tilde v(\theta,e)\bigr)
= \alpha(\theta) + \beta(\theta)\log(1+e),
\end{align}

where $\sigma^{-1}$ is the logit link function, $\alpha(\theta)$ is the baseline performance on the logit scale, and $\beta(\theta) > 0$ is the scaling coefficient governing how efficiently extra compute improves performance.  
If $\alpha(\theta)$ and $\beta(\theta)$ are both weakly increasing in $\theta$, this specification satisfies all conditions in \Cref{ass:post-effort-score} (C1–C3).

We also provide empirical evidence for Assumption~\ref{ass:post-effort-score} using controlled post-training experiments within a single model family. Because latent capability $\theta$ and prior benchmark-specific post-training effort are unobserved, we restrict attention to \textit{Qwen2.5} models~\citep{Yang2024Qwen25TR}. 
Within this family, we use model size as a proxy for capability $\theta$, and we treat the models' pre-existing post-training as approximately the same across sizes. 
We then apply additional benchmark-specific post-training to each model on the benchmark's training set, plotting performance on the validation set as a function of added effort $e$ (measured by the incremental number of training steps).

Figure~\ref{fig:traj} (left) shows the results on \texttt{Winogrande}, a large-scale commonsense pronoun resolution benchmark introduced by \citeauthor{winogrande}, while results for the other eight benchmarks are in \Cref{app:empirical}.
For each model size, we fit the generalized power-law specification in \Cref{eq:log-score}.
The fitted curves track the observed points closely and are consistent with \Cref{ass:post-effort-score}. %\guanhua{How to refer to these?}
First, at any fixed $e$, larger models achieve better performance (monotonicity in $\theta$, C1).
Second, gains from additional effort diminish and scores approach a plateau (concavity and saturation in $e$, C2).
Third, the horizontal distance between curves---interpretable as the extra effort required for a smaller model to match a larger model’s target score---does not shrink at higher target accuracies (increasing effort gaps, C3). 
For instance, the implied effort gap between Qwen2.5-3B and Qwen2.5-7B is larger at \(88\%\) accuracy than at \(86\%\), as illustrated by the dashed guides in Figure~\ref{fig:traj} (left).

\subsection{Order Properties of Equilibrium Profiles}
Our first result shows that if the induced follower game admits a pure-strategy equilibrium $\textbf{e}^*$, the resulting leaderboard ordering must respect the latent capability ordering (\Cref{prop:order-preserve}). Intuitively, this means that if all developers choose benchmark-specific post-training optimally, strategic fine-tuning alone cannot cause a lower-capability model to strictly outrank a higher-capability one:

\begin{proposition}
\label{prop:order-preserve}
Under \Cref{ass:cost} and \Cref{ass:post-effort-score}, fix any tune-before-test adjustment level $\etbt\geq 0$.
If the follower game admits a pure-Nash equilibrium $\textbf{e}^*$, then for any $i,j$, 
$$\theta_i>\theta_j \;\Rightarrow\; v(\theta_i,\etbt+e_i^*) \ge v(\theta_j,\etbt+e_j^*).$$
In particular, post-effort scores at equilibrium preserve the latent capability ordering up to ties.
\end{proposition}

\begin{proof}[Proof sketch]
Fix any $\etbt\ge 0$ and suppose for contradiction that there exist $i,j$ with $\theta_i>\theta_j$ but
$v(\theta_i,\etbt+e^*_i)<v(\theta_j,\etbt+e^*_j)$ at a pure-strategy Nash equilibrium $\textbf{e}^*$.
Let $\tilde e_i$ (resp.\ $\tilde e_j$) be the minimal additional effort for agent $i$ (resp.\ $j$) to reach the other
agent’s equilibrium score.
By the single-crossing condition (C3), the higher-capability agent’s effort advantage is larger at higher target scores,
which implies $\tilde e_i-e_i^*\le e_j^*-\tilde e_j$. By convexity of $c$, this yields
$c(\tilde e_i)-c(e_i^*)\le c(e_j^*)-c(\tilde e_j)$. On the other hand, Nash optimality implies that any profitable rank improvement must be outweighed by the corresponding
increase in cost. Since $v_i^*<v_j^*$, agent $i$ can strictly overtake $j$ by increasing its score to $v_j^*+\varepsilon$
for arbitrarily small $\varepsilon>0$, which yields a strict rank improvement under any tie-breaking rule. By monotonicity
of rank-based rewards, the reward gain from overtaking $j$ is at least the reward loss suffered by $j$, implying the
reverse inequality $c(\tilde e_i)-c(e_i^*)\ge c(e_j^*)-c(\tilde e_j)$.
This means $c(\tilde e_i)-c(e_i^*) =  c(e_j^*)-c(\tilde e_j)$, a contradiction to $v(\theta_i,\etbt+e_i^*)< v(\theta_j,\etbt+e_j^*)$. Thus we have $v(\theta_i,\etbt+e_i^*)\ge v(\theta_j,\etbt+e_j^*)$.

\end{proof}

\Cref{prop:order-preserve} offers reassuring news: if the leaderboard competition ever settles, it will settle in the \emph{right} order. However, this positive result comes with an important caveat. 
As we show next, such equilibria need not exist at all. 
In particular, if the reward gap between some adjacent ranks is large relative to the cost of overtaking, namely, if there exists $r\in \{2, \cdots, n\}$ with $R_{r-1} -R_r$ sufficiently large, the induced follower game fails to admit any pure-strategy Nash equilibrium. 
In this regime, developers face persistent incentives to “just overtake” nearby competitors, leading to arms-race dynamics rather than convergence.

\subsection{Existence of Pure-Nash Equilibrium}

Fixing a TbT $\etbt$, we analyze the induced follower game among developers and characterize when it admits a pure-strategy Nash equilibrium, and when it does not. %Throughout, w.l.o.g., we index models by decreasing capability, $\theta_1 > \theta_2 > \cdots > \theta_n$.

\paragraph{All-Zero additional effort profile.}
A natural reference point is the all-zero additional-effort profile $\textbf{e}= (0, \cdots, 0)$, where model $i$ attains the baseline post-TbT score
\[
s_i(\etbt)=v(\theta_i,\etbt).
\]
By monotonicity in capability (C1), we have $s_1(\etbt) > \dots > s_n(\etbt)$.
Because rewards depend only on rank and effort costs are nondecreasing, any profitable deviation must be an \emph{overtaking move}, and the cheapest such deviation is to overtake the adjacent competitor directly above:

\begin{definition}[Just-Overtake Effort at TbT Level $\etbt$]
\label{def:overtake-effort}
Fix a TbT level $\etbt\ge 0$.
For any rank $r \in \{2, \cdots, n\}$, define the \emph{``just-overtake ''effort} as
\begin{align*}
    {e}_r^\req (\etbt):=
\inf\Bigl\{e\ge 0:\ v(\theta_r,\etbt+e)>s_{r-1}(\etbt)\Bigr\}.
\end{align*}
\end{definition}

% Intuitively, $\hat{e}_r(\etbt)$ is the smallest amount of \emph{additional} benchmark-specific training that agent $r$
% must exert to beat the baseline score of the agent immediately above it.
Equivalently, ${e}_r^\req(\etbt)$ measures how difficult it is for the model designer $r$ to ``climb'' from rank $r$ to rank $r-1$ starting from any common TbT baseline.
When $e^\req_r(\etbt)$ is small, even minor fine-tuning can change rank, creating strong overtaking incentives. When it is large, small investments are unlikely to affect the leaderboard ordering. The following proposition precisely states when the all-zero effort profile is a pure Nash equilibrium strategy of the induced follower game:

% \begin{proposition}
% \label{prop:zero-eq-tbt}
% Fix any $\etbt\ge 0$. The all-zero profile $e^*=(0,\dots,0)$ is a PNE if and only if for every rank $r\in \{2, \cdots n\}$,
% \begin{align}
% \label{eq:zero-eq-condition}
%    c \Big({e}^\req_r\big({\etbt}\big)\Big)\ \ge\ R_{r-1}-R_{r}. 
% \end{align}

% In words, with TbT level $\etbt$, no agent can profitably `just overtake'' the model immediately above it: the minimal cost required to surpass the adjacent baseline score exceeds the corresponding reward gap.
% \end{proposition}

\begin{proposition}[Zero-effort equilibrium condition]
\label{prop:zero-eq-tbt}
Fix any ${\etbt}\ge 0$. The all-zero profile $\textbf{e}= (0, \cdots, 0)$ is a PNE if and only if, for every
$r\in\{2,\ldots,n\}$,
\begin{equation}
\label{eq:zero-eq-condition}
c\!\left(e_r^{\req}({\etbt})\right)\ \ge\ R_{r-1}-R_r.
\end{equation}
In words, with TbT level $\etbt$, no agent can profitably ``just overtake'' the model immediately above it: the minimal cost required to surpass the adjacent baseline score exceeds the corresponding reward gap.
\end{proposition}

\Cref{prop:zero-eq-tbt} further implies that if the adjacent ``just-overtake'' condition fails for some model designer, then no pure-strategy equilibrium exists:

\begin{theorem}[Equilibrium structure and nonexistence]
\label{thm:no-equilibrium}
Fix any $\etbt\ge 0$, and suppose ties are broken deterministically in a way that does not favor the lower-capability model \footnote{Under random tie-breaking, a lower-capability model may benefit from exerting positive effort to enter a tie lottery, so the conclusion that every PNE is all-zero need not hold without further qualification. This is mainly a technical issue: exact ties are rare in practice because leaderboard scores are usually reported at high precision or resolved by fixed secondary rules.}. Then any PNE, if it exists, must be the all-zero profile $\mathbf{e}= (0, \cdots, 0)$.
Moreover, if there exists $r\in\{2,\ldots,n\}$ such that
\begin{equation}
\label{eq:zero-eq-not-hold-condition}
c\!\left(e_r^{\req}(\etbt)\right)<R_{r-1}-R_r,
\end{equation}
then the induced follower game admits no PNE.
\end{theorem}

\begin{proof}[Proof sketch]
Suppose a PNE $\mathbf{e}^*$ exists. By \Cref{prop:order-preserve}, post-effort scores are weakly ordered by capability at equilibrium. Since ties are broken deterministically against the lower-capability model, model $n$ remains last-ranked whenever it only ties a higher-capability model. Thus, it cannot gain from positive effort that does not strictly overtake model $n-1$, and it weakly improves by setting $e_n=0$. Hence $e_n^*=0$. Given $e_n^*=0$, model $n-1$ strictly outranks model $n$ at zero effort by (C1) and cannot lose rank by reducing effort, so $e_{n-1}^*=0$. Repeating this argument inductively yields $e_i^*=0$ for all $i$. Therefore, any PNE must be the all-zero profile $\mathbf{e}= (0, \cdots, 0)$.

If $c(e_r^{\req}(\etbt))<R_{r-1}-R_r$ for some $r\in\{2,\ldots,n\}$, then at $\mathbf{e}=(0, \cdots, 0)$ model $r$ can profitably deviate by exerting effort $e_r^{\req}(\etbt)$, thereby overtaking model $r-1$ and increasing its reward by $R_{r-1}-R_r$ at lower cost. Hence $\mathbf{e}=(0, \cdots, 0)$ is not a PNE. Since any PNE must be all-zero, no PNE exists.
\end{proof}

The theorem shows that incentives are governed locally by the adjacent reward gap $R_{r-1} - R_r$: moving from rank $r$ to rank $r-1$ is profitable only if this gain is large enough relative to the corresponding overtaking cost. This gives a simple interpretation of common reward schemes. In winner-take-all schemes, incentives are concentrated at the top rank; in top-k schemes, they are concentrated near the cutoff into the rewarded set. By contrast, under a smoothly decaying reward scheme, adjacent reward gaps are spread more evenly across ranks and may all be small. In that case, no single rank improvement creates a large reward jump, so overtaking incentives are correspondingly weaker throughout the leaderboard.

Our analysis also suggests that the main failure mode of leaderboards is not that they converge to a \emph{stable but incorrect}
ranking. 
Instead, leaderboards can fail because they may incentivize \emph{continuous competitive fine-tuning}: developers are motivated to repeatedly invest in benchmark-specific optimization simply to gain (or defend) a small rank advantage (\Cref{eq:zero-eq-not-hold-condition}).

This regime is particularly relevant when benchmarks provide limited separation between models, so that even small
performance gains translate into meaningful rank changes.
Such situations can arise on \emph{saturated} benchmarks where frontier models perform similarly (e.g., MMLU~\citep{hendrycks2020measuring} or HellaSwag~\citep{zellers2019hellaswag}), as well as on \emph{very difficult or
early-stage} benchmarks where all models perform poorly and remain tightly clustered (e.g., Humanity's Last Exam~\citep{phan2025humanity}).

%% file: Sections/tune-before-test.tex
\section{Tune-before-Test Aligns Benchmark Incentives}
\label{sec:tbt}
We now turn to the leaderboard designer's perspective and ask: if small reward gaps can preclude equilibrium in the induced follower game, how should the designer choose $\etbt$ to stabilize incentives? At a high level, tune-before-test shifts the leaderboard's ``operating point'' by applying the same amount of benchmark-specific training to all models, so the baseline scores become
$s_i({\etbt}) = v(\theta_i,\etbt)$. This has two conceptually distinct effects.

By moving models closer to their benchmark-specific performance limits, TbT \emph{directly} reduces the leaderboard's sensitivity to who prepared more aggressively from scratch. In the limit as models approach their saturation levels $v^\infty(\theta)=\lim_{e\to\infty}v(\theta,e)$, rankings depend only on capability. Under monotonicity in capability (C1), saturation yields a capability-consistent ordering:

\begin{proposition}[Rank preservation at saturation]
\label{prop:saturation-rank-preservation}
Suppose the post-effort score $v(\theta,e)$ satisfies
$\partial_\theta v(\theta,e) > 0$ (C1).
Then the saturated score
$v^\infty(\theta) := \lim_{e\to\infty} v(\theta,e)$ is increasing in $\theta$.
\end{proposition}

More importantly, TbT also weakens the incentives for strategic post-training.
Since $v(\theta,e)$ exhibits diminishing returns in effort (C2), increasing $\etbt$ pushes all models into a regime where further improvements are harder.

\subsection{TbT Increases Leaderboard Climbing Cost}
To formalize this incentive effect, recall $e_r^\req(\etbt)$, the minimal additional effort required for model $r$ to overtake the adjacent competitor at baseline $\etbt$ (\Cref{def:overtake-effort}). This precisely captures the marginal difficulty of climbing the leaderboard. The next lemma shows that increasing the TbT baseline monotonically raises this overtaking cost:

\begin{lemma}[TbT Increases the Leaderboard Climbing Cost]
\label{lemma:tbt-increase-cost-of-climbing}
For any $\etbt\ge0$ and $r\in\{2,\dots,n\}$, the minimal just-overtake effort $e^\req_r(\etbt)$, and thus also the corresponding cost $c(e^\req_r(\etbt))$, is non-decreasing in $\etbt$.
\end{lemma}

\Cref{fig:traj} (right) plots the minimal just-overtake effort across all models $\min_{r\in \{2,\ldots,n\}}{e}^\req_r(\etbt)$ ($y$-axis), as a function of $\etbt$ ($x$-axis).
We estimate ${e}^\req_r(\etbt)$ using the fitted post-training trajectories shown in the left panel. For example, for model Qwen2.5-7B at $\etbt=500$, we first compute the accuracy of Qwen2.5-14B at $\etbt=500$, 91.0\%. We then calculate how much additional training data Qwen2.5-7B requires to reach 91.0\% based on the fitted curve, which is 5,890.
We repeat this for all models (except for the best model Qwen2.5-14B) and plot the minimum across all rank $r\geq 2$.
The resulting curve shows that leaderboard climbing becomes rapidly more difficult as $\etbt$ grows. At $\etbt=0$, only 18 additional steps are needed to change the ranking, whereas at $\etbt=3,000$, the required effort rises to $384,668$ steps.

\subsection{TbT Restores Equilibrium Existence}
Since increasing TbT raises the effort required to overtake higher-ranked models, the all-zero equilibrium condition in \Cref{prop:zero-eq-tbt} becomes easier to satisfy at higher TbT levels:

\begin{proposition}
\label{prop:tbt-monotonic-incentive}
Given two tune-before-test adjustment levels $0\le \Delta^{tbt}_1 \le \Delta^{tbt}_2$, if the all-zero effort profile $\textbf{e}=(0,\dots,0)$ is a PNE under $\Delta^{tbt}_1$, then it is also a PNE under $\Delta^{tbt}_2$.
\end{proposition}

Thus, TbT has a monotone stabilizing effect on post-training incentives: once no model developer finds it profitable to exert additional effort, this remains true under any larger TbT intervention.

\subsection{How Much TbT is Enough?}
In practice, a benchmark evaluator may not have enough compute to apply a large TbT adjustment to every submitted model, due to resource, latency, or cost constraints. This raises a natural question: what is the smallest TbT level that eliminates incentives for further strategic fine-tuning? This motivates us to view TbT as an incentive-control parameter with a corresponding threshold.

\begin{definition}[Stabilizing TbT Threshold]
\label{def:tbt-threshold}
Define the \emph{stabilizing TbT threshold} as
\[
{\etbt}^\star
:=
\inf\Bigl\{
\etbt\ge 0: 
c\big(e^\req_r(\etbt)\big)\ \ge\ R_{r-1}-R_r,
\forall r \in \{2, \cdots, n\}
\Bigr\},
\]

\noindent that is, the smallest TbT adjustment for which no model can profitably overtake the model directly above it at the baseline level.
\end{definition}

At ${\etbt}^\star$, the cost of overtaking the immediate neighbor weakly exceeds the corresponding reward gain for every adjacent pair. By \Cref{prop:tbt-monotonic-incentive}, once this condition holds, it continues to hold for all larger TbT levels. Hence, any $\etbt \ge {\etbt}^\star$ eliminates profitable overtaking incentives:

\begin{corollary}
\label{cor:tbt-stabilizes-incentives}
For any $\etbt \ge {\etbt}^\star$, the induced follower game admits a PNE in which all model developers choose zero additional effort, i.e., $\textbf{e}^*=(0,\dots,0)$. In particular, no model developer engages in strategic post-training.
\end{corollary}

Under the leaderboard designer objective in \Cref{def:leader-utility}, ${\etbt}^\star$ is also the optimal TbT choice: since it is the smallest intervention that induces a follower equilibrium preserving the true capability ordering, it therefore achieves the designer’s ranking objective at minimum cost.

\subsection{Case Study Under Generalized Power-Law}
To build intuition for the scale of the stabilizing TbT level ${\etbt}^\star$, we consider the generalized scaling law introduced above:
\[
\sigma^{-1}\!\bigl(\tilde v(\theta,e)\bigr)=\alpha(\theta)+\beta(\theta)\log(1+e),
\]
and assume effort costs satisfy $c(e)\ge \kappa e$ for some constant $\kappa>0$, so that effort is costly at least linearly. We focus on the regime in which overtaking remains feasible. For an adjacent pair $(r-1,r)$, define $\rho_r := \frac{R_{r-1}-R_r}{\kappa}$,
which expresses the reward gap in units of effort. Let $e_r^{\req}(0)$ denote the minimal effort required for model $r$ to catch up to model $r-1$ when $\etbt=0$. We then define
\[
\lambda_r := \frac{\rho_r}{e_r^{\req}(0)},
\]
which measures the reward gap relative to the baseline catch-up difficulty.

The next proposition shows that, under this scaling law, the stabilizing TbT level for pair $(r-1,r)$ grows polynomially in $\lambda_r$, with exponent
$\gamma_r := \frac{\beta(\theta_r)}{\beta(\theta_{r-1})}\le 1$ interpreted as a relative learning-rate ratio between the two models.

\begin{proposition}[Stabilizing TbT under generalized scaling]
\label{prop:powerlaw-linear-cost-stabilizing-tbt}
Under the generalized scaling law, and the cost satisfies $c(e)\geq \kappa e$,
there exists a stabilizing TbT threshold ${\etbt}_r^\star$ such that any ${\etbt}\ge {\etbt}_r^\star$ eliminates profitable overtaking deviations from rank $r$ to rank $r-1$. Moreover, in the regime where catch-up remains feasible,
\[
\etbtrstar = O\!\bigl(\lambda_r^{\gamma_r}\bigr),
\]
up to a constant factor depending only on $e_r^{\req}(0)$.
\end{proposition}

The exponent $\gamma_r$ admits a natural interpretation as a relative learning-rate ratio. When $\gamma_r$ is small, the stronger model benefits more from additional effort than the weaker one, so the ranking becomes easier to stabilize.
Conversely, when $\gamma_r\approx 1$, the two models benefit similarly from effort, and a larger TbT baseline is needed to deter overtaking.
Thus, the required TbT level grows polynomially in the effective incentive $\lambda_r$, and the growth is slower when the effort-response gap between the two models is larger. 

Finally, to stabilize the full leaderboard, it suffices to choose
\[
{\etbt}^\star := \max_{r\in \{2, \cdots, n\}} \etbtrstar,
\]
so that the global threshold is determined by the hardest adjacent pair to stabilize.

\paragraph{A concrete example.}
We estimate the baseline catch-up difficulty $e^\req_r(0)$ and the learning-rate
ratio $\gamma_r$ from the fitted trajectories in \Cref{fig:traj}. 
Among all adjacent
pairs, the hardest-to-stabilize pair maximizes ${\etbt}^\star=\max_r \Delta_r^{tbt\star}$.
For this pair, we estimate $e^\req_r(0)\approx 13.7$ and $\gamma_r\approx 0.33$.
Suppose the corresponding reward gap translates to $\rho_r=1000$ units of effort,
so the effective incentive is
$\lambda_r=\frac{\rho_r}{e^\req_r(0)}\approx \frac{1000}{13.7}\approx 73$.
Using Proposition~\ref{prop:powerlaw-linear-cost-stabilizing-tbt}, a stabilizing
TbT level is
\[
\Delta_r^{tbt\star}
\approx
\left(\frac{\rho_r}{e^\req_r(0)}\right)^{\gamma_r}
=
\left(\frac{1000}{13.7}\right)^{0.33}
\approx 4,
\]

That is, adding roughly 4 units of baseline TbT effort is already enough to deter overtaking for this hardest-to-stabilize adjacent pair. Detailed estimation is deferred to the appendix.

%% file: Sections/conclusion.tex
\section{Conclusion, Limitations and Future Work}
We study benchmarking as a mechanism design problem and show that the resulting competition among model developers can induce persistent ``just-overtake'' incentives and may prevent any Nash equilibrium from existing.
We demonstrate that tune-before-test acts as an effective incentive-control lever: by pushing models into a diminishing-returns regime, TbT raises the marginal cost of further benchmark-specific post-training and amplifies the effort required to overtake nearby competitors. 
Our results show that even a small amount of TbT can have a disproportionately large stabilizing effect, which we demonstrate both empirically and theoretically. 
In practice, leaderboard designers could estimate the minimal stabilizing TbT baseline ${\etbt}^\star$ using historical score gaps and the observed fine-tuning effort required to change rankings. 

% limitation on the model
These insights, however, rest on several simplifying assumptions. In particular, our model abstracts from many features of real benchmarking environments in order to isolate the core incentive structure induced by rank-based evaluation. A natural next step is to understand how these incentives interact with features that are common in practice, such as noisy evaluations, uncertainty about competitors, and more flexible budget or cost structures. Extending the framework in these directions would help clarify the robustness of our insights in more realistic benchmarking environments.

% limitation on TbT as an intervention method
At the same time, TbT is not costless: it requires additional evaluation resources and may blur the distinction between a model’s underlying generalization ability and its capacity to adapt during evaluation. More broadly, these trade-offs highlight that benchmark design does not merely determine how models are tested, but also shapes the incentives for how they are developed. TbT should therefore be understood as a useful but imperfect design lever, rather than a complete solution to strategic behavior in evaluation.

%% file: Sections/acknowledgement.tex
\section*{Acknowledgement}
We thank Jiduan Wu, Vivian Nastl,
Ana-Andreea Stoica, André Cruz,
Ricardo Dominguez-Olmedo, Florian Dorner, Dingfan Chen, Tom Sühr,
Nikhil Chandak, Mila Gorecki, Mina Remeli, Omri Ben-Dov, and anonymous reviewers from the ICLR AIMS workshop for the helpful discussions and feedback on this work.

%% file: Sections/appendix.tex
\section{Primary Notation Table}

\begin{table}[!ht]
    \centering
    \small
    \setlength{\tabcolsep}{5pt}
    \begin{tabularx}{\linewidth}{c >{\raggedright\arraybackslash}X}
        \toprule
        Symbol & Description \\
        \midrule
        $\theta \in \mathbbm{R}_{\ge 0}$ & one-dimensional latent model capability \\
        $e \in \mathbbm{R}_{\ge 0}$ & benchmark-specific training effort \\
        $v(\theta,e)\in[0,1]$ & post-effort score (benchmark performance) \\
        $c(e)$ & cost of exerting effort $e$ \\
        $(R_j)_{j=1}^n$ & rank-based reward scheme: the model ranked $j$ receives reward $R_j$ \\
        $\etbt \in \mathbbm{R}_{\ge 0}$ & tune-before-test (TbT) baseline chosen by the leaderboard designer \\
        $s_r(\etbt):=v(\theta_r,\etbt)$ & baseline score of model $r$ under TbT level $\etbt$ \\
        $e^\req(s;\theta)$ & minimal effort required for a model with capability $\theta$ to reach score $s$, i.e.,
        $e^\req(s;\theta):=\inf\{e\ge 0:\ v(\theta,e)\ge s\}$ \\
        $e_r^\req(\etbt)$ & minimal additional effort required for model $r$ to overtake model $r-1$ at TbT level $\etbt$, i.e.,
        $e_r^\req(\etbt):=\inf\{e\ge 0:\ v(\theta_r,\etbt+e)> s_{r-1}(\etbt)\}$ \\
        $\sigma:\mathbbm{R}\to[0,1]$ & logit link function \\
        $L(\theta)$ & lower attainable performance level for a model with capability $\theta$ \\
        $U(\theta)$ & upper attainable performance level for a model with capability $\theta$ \\
        $\tilde v(\theta,e):=\frac{v(\theta,e)-L(\theta)}{U(\theta)-L(\theta)}\in[0,1]$ & normalized post-effort score \\
        $\alpha(\theta)$ & logit of the baseline performance when $e=0$ \\
        $\beta(\theta)>0$ & coefficient governing how effectively extra effort translates into performance \\
        $\kappa > 0$ & constant lower bound on the marginal cost of effort, i.e., $c(e)\ge \kappa e$ \\
        $\rho_r := (R_{r-1}-R_r)/\kappa$ & effective reward gap (in effort units) for model $r$ \\
        $\lambda_r:=\rho_r/e_r^\req(0)$ & effective incentive parameter \\
        $\gamma_r:=\beta(\theta_r)/\beta(\theta_{r-1})\le 1$ & learning-rate ratio for the pair $(r-1,r)$ \\
        $\etbtrstar$ & stabilizing TbT threshold for the adjacent pair $(r-1,r)$ \\
        ${\etbt}^\star:=\max_{r\in\{2,\ldots,n\}} \etbtrstar$ & global stabilizing TbT threshold \\
        \bottomrule
    \end{tabularx}
    \vspace{0.5em}
    \caption{Primary notation.}
    \label{tab:notation_table}
\end{table}

\newpage 

\section{Missing Proof in \Cref{sec:game}}
\paragraph{Equivalence of heterogeneous separable costs to a homogeneous same-family cost}

\begin{proposition}[Equivalence of heterogeneous separable costs to a homogeneous same-family cost]
\label{prop:equivalence-cost-fn}
Fix a common cost shape $c:\mathbb{R}_+\to\mathbb{R}_+$ that is continuous, strictly increasing, and satisfies $c(0)=0$. 
Consider each model $i$ chooses effort $e_i\in\mathbb{R}_+$, is ranked by a score $v(\theta_i,e_i)$, and has a heterogeneous but multiplicatively separable cost
\[
C_i(e_i;\theta_i)=\gamma_i\,c(e_i)\qquad\text{with }\gamma_i>0.
\]
Then there exists an \emph{equivalent} game in which every model developer chooses $z_i\in\mathbb{R}_+$, all model developers share the \emph{same} cost function of the \emph{same form} $c(\cdot)$ (i.e., homogeneous cost $c(z_i)$), and the only change is a type-dependent relabeling of the score:
\[
\hat v(\theta_i,z_i)\;:=\;v\!\Big(\theta_i,\,c^{-1}\!\Big(\frac{c(z_i)}{\gamma_i}\Big)\Big).
\]
Specifically, the map
\[
\Phi_i:\ e_i\ \Leftrightarrow\ z_i:=c^{-1}\!\big(\gamma_i\,c(e_i)\big)
\]
is a bijection for each $i$, and for every profile $(e_j)_j$ with image $(z_j)_j$ we have
\[
R_{\mathrm{rank}(v(\theta_i,e_i))}\;-\;\gamma_i\,c(e_i)
\;=\;
R_{\mathrm{rank}(\hat v(\theta_i,z_i))}\;-\;c(z_i)\quad\text{for all }i.
\]
Consequently, the two games induce identical rankings, payoffs, best responses, and Nash equilibria up to the one-to-one reparametrization.
\end{proposition}

\begin{proof}
Because $c$ is strictly increasing, $c^{-1}$ exists and is strictly increasing. 
For each $i$, define the type-wise bijection $\Phi_i(e):=c^{-1}(\gamma_i\,c(e))$, with inverse
\[
\Phi_i^{-1}(z)\;=\;c^{-1}\!\Big(\frac{c(z)}{\gamma_i}\Big).
\]
Under $\Phi_i$, the heterogeneous cost transforms as
\[
\gamma_i\,c(e_i)\;=\;c\!\big(\Phi_i(e_i)\big)\;=\;c(z_i),
\]
so all model developers share the same cost function $c(\cdot)$ in the $z$-parametrization.
Define $\hat v(\theta_i,z_i):=v\!\big(\theta_i,\Phi_i^{-1}(z_i)\big)=v\!\big(\theta_i,c^{-1}(c(z_i)/\gamma_i)\big)$. 
Then for any profile $(e_j)_j$ with the reparametrization variable $(z_j)_j$, we have:
\[
R_{\rank(v(\theta_i,e_i))}-\gamma_i c(e_i)
\;=\;
R_{\rank\big(v(\theta_i,\Phi_i^{-1}(z_i))\big)}-c(z_i)
\;=\;
R_{\rank\big(\hat v(\theta_i,z_i)\big)}-c(z_i),
\]
so the utilities are preserved. Hence, the best responses and equilibria correspond one-to-one via $\Phi$, establishing the equivalence between the two settings.
\end{proof}

\begin{figure}
    \centering
    \includegraphics[width=0.8\linewidth]{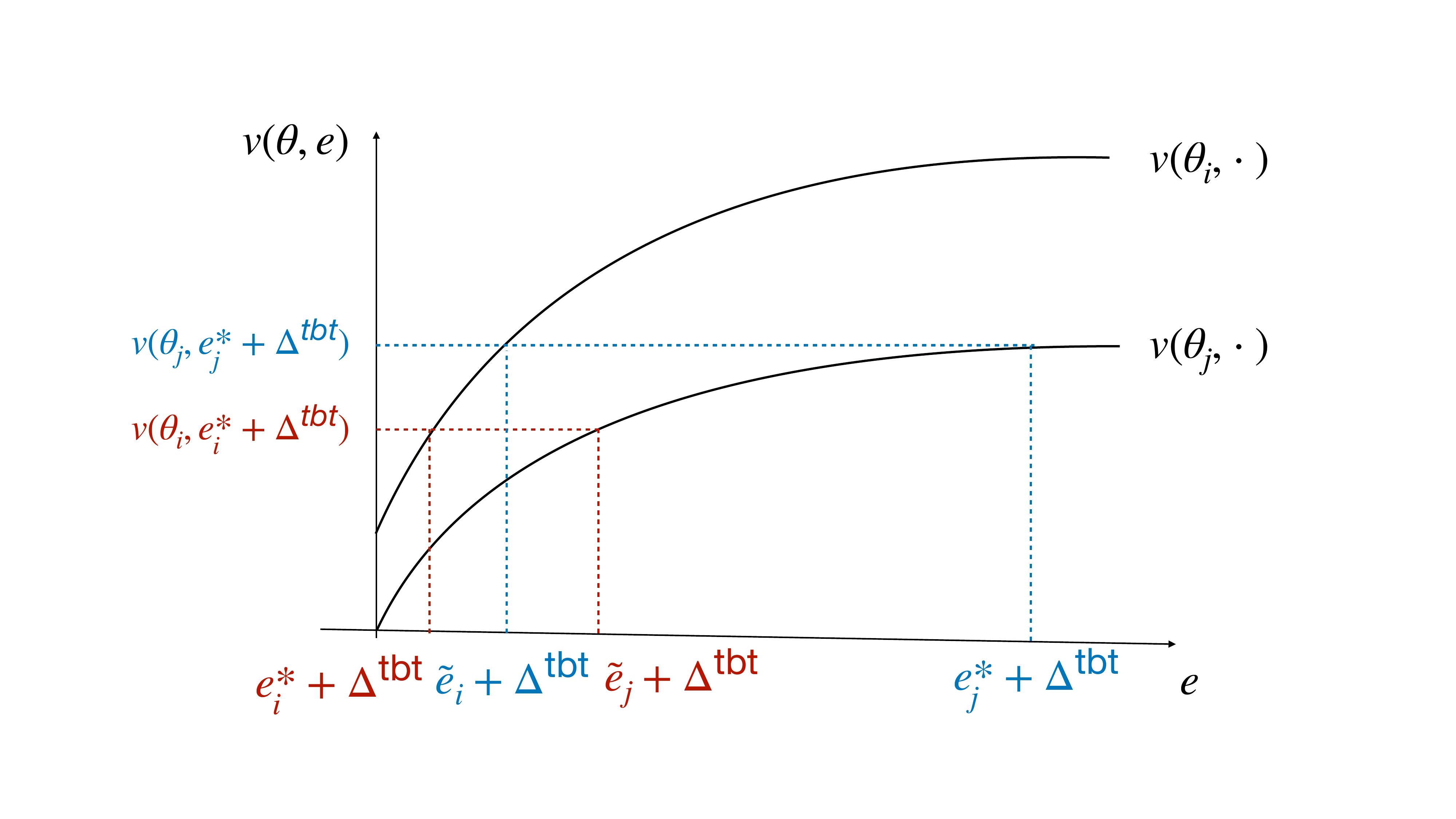}
    \caption{Illustration of the notations and functions used in the proof of \Cref{prop:order-preserve}. The post-effort score function \( v(\theta, e) \) satisfies the conditions in \Cref{ass:post-effort-score}. Here, \( e_i^* \) and \( e_j^* \) denote the equilibrium efforts of model developers \( i \) and \( j \), respectively. The counterfactual efforts \( \tilde{e}_i \) and \( \tilde{e}_j \) are defined as the efforts each model developer would need to match the other model developer's post-effort score given a TbT effort level $\etbt$, i.e., $v(\theta_i, \tilde{e}_i + \etbt ) = v(\theta_j, e_j^* + \etbt), v(\theta_j, \tilde{e}_j + \etbt) = v(\theta_i, e_i^* + \etbt).$ }
    \label{fig:notation-demonstration}
\end{figure}

% \section{Leaderboard designer's utility}
% \label{app:leader-utility}
% A minimal way to capture this objective is through a \emph{ranking correctness} criterion:

% \begin{definition} (Leaderboard Designer's Utility)
% \label{def:leader-utility}
% Given a tune-before-test level $\etbt$ and a follower's effort profile $e = (e_1, \cdots, e_n)$ inducing post-effort scores $v_i = v(\theta_i, e_i + \etbt)$, the leaderboard designer's utility is:
% \begin{align}
% \label{eqn:leader-utility}
%     U^L(\etbt; e) = R^L \cdot \mathbbm{I}[\rank(v_i) = \rank(\theta_i), \forall i] - n\cdot c^L(\etbt),
% \end{align}
% where $R^L > 0$ is the reward from achieving a capability-consistent ranking, and $c^L(\etbt)$ is the per-model cost of applying a tune-before-test.
% \end{definition}

% This binary formulation captures the designer’s core concern: whether the leaderboard correctly orders models by latent capability. One could alternatively model the designer’s objective more smoothly using a rank-correlation metric such as Kendall’s $\tau$ \citep{kendall1938rank}, which rewards partial agreement between score-based and capability rankings, minus the cost of tune-before-test adjustment.

\section{Proofs for \Cref{sec:equilibrium}}

\paragraph{Proof for \Cref{prop:order-preserve}}

\begin{proof}
\label{proof:rank-preservation}
Fix any $\etbt\ge 0$ and let $\textbf{e}^*$ be a PNE of the induced follower game.
Suppose for contradiction that there exist models $i,j$ such that $\theta_i>\theta_j$ but
\begin{equation}
\label{eq:misorder-assump}
v_i^* \;:=\; v(\theta_i,\etbt+e_i^*) \;<\; v(\theta_j,\etbt+e_j^*) \;=:\; v_j^* .
\end{equation}

\paragraph{Counterfactual efforts.}
For any target score $s\in[0,1]$, recall the minimal additional effort (beyond $\etbt$) needed for capability $\theta$ to
reach $s$ is:
\[
e^\req(s;\theta, \etbt):=\inf\{e\ge 0:\ v(\theta,\etbt+e)\ge s\}.
\]
Define the counterfactual efforts
\[
\tilde e_i:= e^\req (v_j^*;\theta_i, \etbt),
\qquad
\tilde e_j:= e^\req(v_i^*;\theta_j, \etbt).
\]
Then by definition,
\[
v(\theta_i,\etbt+\tilde e_i)\ge v_j^*,
\qquad
v(\theta_j,\etbt+\tilde e_j)\ge v_i^*.
\]
Moreover, since $v_j^*>v_i^*$ and $v(\theta,\cdot)$ is nondecreasing in effort (C2),
we have $\tilde e_i>e_i^*$ and $\tilde e_j\le e_j^*$.

% \paragraph{Counterfactual efforts.}
% For any target score $s\in[0,1]$, define the minimal additional effort (beyond $\etbt$) required for capability $\theta$ to reach $s$:
% \[
% e^{\req}(s;\theta)
% :=\inf\{e\ge 0:\ v(\theta,\etbt+e)\ge s\}.
% \]
% Equivalently, write $\hat e^{\etbt}(s;\theta):=e^{\req}(s;\theta)$.
% Define the counterfactual efforts
% \[
% \tilde e_i:=\hat e^{\etbt}(v_j^*;\theta_i),
% \qquad
% \tilde e_j:=\hat e^{\etbt}(v_i^*;\theta_j).
% \]
% Then by definition,
% \[
% v(\theta_i,\etbt+\tilde e_i)\ge v_j^*,
% \qquad
% v(\theta_j,\etbt+\tilde e_j)\ge v_i^*.
% \]
% Moreover, since $v_j^*>v_i^*$ and $v(\theta,\cdot)$ is nondecreasing in effort (C2),
% we have $\tilde e_i>e_i^*$ and $\tilde e_j\le e_j^*$.

\paragraph{Effort-gap comparison.}
Let
\[
\Delta_{\theta_i,\theta_j}(s)
:= e^{\req}(s;\theta_j,\etbt)- e^{\req}(s;\theta_i,\etbt)
\]
denote the effort gap required for the lower-capability model $j$ to achieve score $s$ relative to the higher-capability model $i$ (at TbT level $\etbt$).
By (C3), $\Delta_{\theta_i,\theta_j}(s)$ is (weakly) nondecreasing in $s$.
Applying this with $s=v_j^*>v_i^*$ gives
\[
\Delta_{\theta_i,\theta_j}(v_j^*) \;\ge\; \Delta_{\theta_i,\theta_j}(v_i^*),
\]
i.e.,
\[
e^{\req}(v_j^*;\theta_j,\etbt)-e^{\req}(v_j^*;\theta_i,\etbt)
\;\ge\;
e^{\req}(v_i^*;\theta_j,\etbt)-e^{\req}(v_i^*;\theta_i,\etbt).
\]
Using $e^{\req}(v_i^*;\theta_i,\etbt)=e_i^*$ and $e^{\req}(v_j^*;\theta_j,\etbt)=e_j^*$, and the
definitions $\tilde e_i:=e^{\req}(v_j^*;\theta_i,\etbt)$ and $\tilde e_j:=e^{\req}(v_i^*;\theta_j,\etbt)$, this becomes
\begin{equation}
\label{eq:effort-ineq}
\tilde e_i-e_i^* \;\le\; e_j^*-\tilde e_j.
\end{equation}

By convexity and monotonicity of $c$ (\Cref{ass:cost}),
\Cref{eq:effort-ineq} implies
\begin{equation}
\label{eq:cost-ineq}
c(\tilde e_i)-c(e_i^*) \;\le\; c(e_j^*)-c(\tilde e_j).
\end{equation}

\paragraph{Nash inequalities.}
Since $\textbf{e}^*$ is a PNE, neither model can gain by deviating unilaterally. Let
$\textbf{v}^*=(v_1^*,\dots,v_n^*)$ be the equilibrium score profile and let $r_k(\cdot)$
denote the rank of model $k$ under a given score profile (with an arbitrary but fixed
deterministic tie-breaking rule). Then
\begin{align*}
R_{r_i(\textbf{v}^*)}-c(e_i^*)
&\ge
R_{r_i\!\bigl(v(\theta_i,\etbt+\tilde e_i),\,\textbf{v}_{-i}^*\bigr)}-c(\tilde e_i),\\
R_{r_j(\textbf{v}^*)}-c(e_j^*)
&\ge
R_{r_j\!\bigl(v(\theta_j,\etbt+\tilde e_j),\,\textbf{v}_{-j}^*\bigr)}-c(\tilde e_j).
\end{align*}
Rearranging terms gives
\begin{align}
\label{eq:nash-i}
R_{r_i\!\bigl(v(\theta_i,\etbt+\tilde e_i),\,\textbf{v}_{-i}^*\bigr)}-R_{r_i(\textbf{v}^*)}
&\le c(\tilde e_i)-c(e_i^*),\\
\label{eq:nash-j}
R_{r_j(\textbf{v}^*)}-R_{r_j\!\bigl(v(\theta_j,\etbt+\tilde e_j),\,\textbf{v}_{-j}^*\bigr)}
&\ge c(e_j^*)-c(\tilde e_j).
\end{align}
Here \Cref{eq:nash-i} is the reward increase when $i$ raises its effort to (at least) match $j$'s score $v_j^*$, and
\Cref{eq:nash-j} is the reward decrease when $j$ lowers its effort to (at least) match $i$'s score $v_i^*$.

\paragraph{Reward comparison via strict overtaking.}
Fix any $\varepsilon>0$, and define the $\varepsilon$-overtake effort for model $i$:
\[
\tilde e_i^{\,\varepsilon}:=\inf\{e\ge 0:\ v(\theta_i,\etbt+e)\ge v_j^*+\varepsilon\}.
\]
Then $v(\theta_i,\etbt+\tilde e_i^{\,\varepsilon})>v_j^*$, so model $i$ strictly outranks model $j$ under \emph{any}
deterministic tie-breaking rule. Hence the reward increase of model $i$ from deviating to $\tilde e_i^{\,\varepsilon}$
is at least the reward decrease of model $j$ when it deviates to achieve $v_i^*$:
\begin{equation}
\label{eq:reward-compare}
R_{r_i\!\bigl(v(\theta_i,\etbt+\tilde e_i^{\,\varepsilon}),\,\textbf{v}_{-i}^*\bigr)}-R_{r_i(\textbf{v}^*)}
\;\ge\;
R_{r_j(\textbf{v}^*)}-R_{r_j\!\bigl(v(\theta_j,\etbt+\tilde e_j),\,\textbf{v}_{-j}^*\bigr)}.
\end{equation}

Combining \eqref{eq:nash-i}, \eqref{eq:nash-j}, and \eqref{eq:reward-compare} yields
\[
c(\tilde e_i^{\,\varepsilon})-c(e_i^*) \;\ge\; c(e_j^*)-c(\tilde e_j).
\]
Letting $\varepsilon\downarrow 0$ and using continuity of $v(\theta,\cdot)$ and $c(\cdot)$ gives
\[
c(\tilde e_i)-c(e_i^*) \;\ge\; c(e_j^*)-c(\tilde e_j).
\]

Together with \eqref{eq:cost-ineq}, we must have
\[
c(\tilde e_i)-c(e_i^*) =  c(e_j^*)-c(\tilde e_j).
\]

However, since $v_i^*<v_j^*$ by assumption, model $i$ can strictly improve its rank by exerting effort slightly above
$\tilde e_i$ (i.e., achieving $v_j^*+\varepsilon$), which yields a \emph{strictly} higher reward whenever rewards are
strictly decreasing in rank. This would force a strict inequality in the Nash bound for $i$, contradicting equality.

Therefore \Cref{eq:misorder-assump} is impossible, and for any $\theta_i>\theta_j$ we must have
\[
v(\theta_i,\etbt+e_i^*) \;\ge\; v(\theta_j,\etbt+e_j^*).
\]
In particular, equilibrium scores are capability-consistent (up to ties).

\end{proof}

\paragraph{Proof for \Cref{prop:zero-eq-tbt}}
\begin{proof}
\label{proof:prop:zero-eq-tbt}
Fix any tune-before-test adjustment level $\etbt\ge 0$.
At the all-zero \emph{additional-effort} profile $\textbf{e}=(0,\dots,0)$, the realized scores are
$s_i(\etbt)=v(\theta_i,\etbt)$ for $i=1,\dots,n$.
Since models are indexed so that $\theta_1>\theta_2>\cdots>\theta_n$ and $v(\theta,e)$ is increasing in $\theta$ (C1),
we have
$s_1(\etbt)> s_2(\etbt)> \cdots > s_n(\etbt)$.
Hence the payoff of model $k$ at this profile is $U_k(0; \textbf{0}_{-k}, \etbt) = R_k$ for $k=1,\dots,n$.

Fix a model $r\in\{2,\dots,n\}$.
Consider any unilateral deviation $e_r'>0$ by model $r$.
Because rewards depend only on ranks and $c$ is nondecreasing, among all deviations that improve model $r$'s rank,
the cheapest such deviation is to \emph{just overtake} some model currently above it, i.e., to move from rank $r$
to some rank $k<r$ by achieving a score strictly greater than $s_k(\etbt)$.
Define the minimal additional effort needed for model $r$ to strictly beat score $s_k(\etbt)$ as
\[
e^{\req}\!\bigl(s_k(\etbt);\theta_r,\etbt\bigr)
:=\inf\{e\ge 0:\ v(\theta_r,\etbt+e)>s_k(\etbt)\}.
\]
Such a deviation yields utility
\[
R_k - c\!\left(e^{\req}\!\bigl(s_k(\etbt);\theta_r,\etbt\bigr)\right).
\]
Therefore, $e_r=0$ is a best response for model $r$ at the all-zero profile if and only if
\[
R_r \;\ge\; R_k - c\!\left(e^{\req}\!\bigl(s_k(\etbt);\theta_r,\etbt\bigr)\right)
\qquad\text{for all } k<r,
\]
equivalently,
\[
c\!\left(e^{\req}\!\bigl(s_k(\etbt);\theta_r,\etbt\bigr)\right)\ \ge\ R_k-R_r
\qquad\text{for all } k<r.
\]
and we shorthand $e^\req_r(\etbt):= e^{\req}\!\bigl(s_{r-1}(\etbt);\theta_r,\etbt\bigr)$. 

Finally, since overtaking a higher-ranked model requires weakly more effort than overtaking the adjacent model $r-1$,
it suffices to check the adjacent deviation $k=r-1$, yielding \Cref{eq:zero-eq-condition}.
\end{proof}

\paragraph{Proof for \Cref{thm:no-equilibrium}}
\begin{proof}
\label{proof:thm:no-equilibrium}
We first show that if a PNE exists, it must be the all-zero additional-effort profile; we then show that under the stated condition, the all-zero profile cannot be a PNE.

Fix $\etbt\geq 0$ and suppose $\textbf{e}^*$ is a PNE.
By \Cref{prop:order-preserve}, post-effort scores are weakly ordered by capability at equilibrium:
if $\theta_i>\theta_j$, then $v(\theta_i,\etbt+e_i^*)\ge v(\theta_j,\etbt+e_j^*)$.
Under deterministic tie-breaking that favors higher capability, this implies that model $n$ (the lowest capability $\theta_n$) is last-ranked at effort level $e^*$. Since at any effort profile, the last-ranked model receives reward $R_n$ regardless of its own effort, thus model $n$ can weakly improve its utility by reducing its effort to zero: its reward cannot decrease, while its cost weakly decreases, since $c$ is nondecreasing and $c(0)=0$.
Therefore, in any PNE, we must have $e_n^*=0$.

Now consider model $n-1$. Given $e_n^*=0$, if model $n-1$ sets effort to zero as well, its score becomes $v(\theta_{n-1},\etbt)$, which is weakly higher than $v(\theta_n,\etbt)$ by the assumption that $v(\theta, e)$ is monotonic in capability (C1).
Under the same tie-breaking rule, model $n-1$ cannot fall below model $n$ by choosing zero effort, and its cost decreases.
Hence $e_{n-1}^*=0$ as well.
Proceeding inductively, we obtain $e_i^*=0$ for all $i$, so any PNE must be $\textbf{e}^*=(0,\dots,0)$.

Next, we show that under the stated condition \cref{eq:zero-eq-not-hold-condition}, all-zero is not a PNE.
At $\textbf{e}=(0,\dots,0)$, baseline scores are $s_i(\etbt)=v(\theta_i,\etbt)$.
If for some $r\in\{2,\dots,n\}$,
\[
c \Big({e}_r^\req \big(\etbt\big)\Big) < R_{r-1}-R_r,
\]
then model $r$ has a profitable deviation, namely that it can exert effort $\hat{e} = {e}^\req_r\big(\etbt\big)$, 
so that $v(\theta_r,\etbt+\hat e) > s_{r-1}(\etbt)$, thereby overtaking model $r-1$ and improving its reward from $R_r$ to at least $R_{r-1}$.
The resulting utility gain is at least $R_{r-1}-R_r$ while the incurred cost is $c(\hat e)$, which is strictly smaller by assumption.
Therefore, the deviation is profitable, so the all-zero profile is not a PNE. As a result, no PNE exists in this case.
\end{proof}

\section{Proof in \Cref{sec:tbt}}
%\paragraph{Statement and Proof for \Cref{prop:saturation-rank-preservation}}
\paragraph{Proof for \Cref{prop:saturation-rank-preservation}}
% \begin{proposition}[Rank preservation at saturation]
% \label{prop:saturation-rank-preservation}
% Suppose the post-effort score $v(\theta,e)$ satisfies
% $\partial_\theta v(\theta,e) > 0$ (C1).
% Then the saturated score
% $v^\infty(\theta) := \lim_{e\to\infty} v(\theta,e)$ is nondecreasing in $\theta$.
% \end{proposition}
\begin{proof}
\label{proof:prop:saturation-homogenous-scaling}
Fix any pair $(i,j)$ with $\theta_i > \theta_j$ and define the score gap at effort level $e\geq 0$ as
$d_{ij}(e) := v(\theta_i,e) - v(\theta_j,e)$. By C1 ($\partial_\theta v > 0$), we have $d_{ij}(e) > 0$ for all $e$. Hence,
\[
v^\infty(\theta_i) - v^\infty(\theta_j)
= \lim_{e\to\infty}\bigl(v(\theta_i,e)-v(\theta_j,e)\bigr)
= \lim_{e\to\infty} d_{ij}(e) > 0,
\]
so $v^\infty(\theta)$ is increasing in $\theta$. 
% If additionally we have nonvanishing asymptotic gap, namely $\liminf_{e\to\infty} d_{ij}(e) > 0$, then $v^\infty(\theta_i) - v^\infty(\theta_j) > 0$, proving strict preservation.
\end{proof}

\paragraph{Proof for \Cref{lemma:tbt-increase-cost-of-climbing}}
\begin{proof}
\label{proof:lemma:tbt-increase-cost-of-climbing}
Fix $r\in\{2,\dots,n\}$ with $\theta_{r-1}>\theta_r$, recall the baseline score for model $r-1$ with TbT adjustment level $\etbt$ is
\[
s_{r-1}(\etbt)\ :=\ v(\theta_{r-1},\etbt).
\]
By C2, $s_{r-1}(\etbt)$ is non-decreasing in $\etbt$.

Let $e^\req(s;\theta):=\inf\{e\ge 0:\ v(\theta,e)\ge s\}$ be the minimal \emph{total} effort to reach score $s$.
Define the (weak) catch-up effort for model $r$ at baseline $\etbt$ by
\[
\delta_r(\etbt)\ :=\  e^\req \!\big(s_{r-1}(\etbt);\theta_r\big)\;-\;\etbt .
\]
% (Indeed, $\etbt+\delta_r(\etbt)$ is exactly the minimal total effort for capability $\theta_r$ to reach
% the score $s_{r-1}(\etbt)$.)

Since $s_{r-1}(\etbt)$ is achieved by capability $\theta_{r-1}$ at effort $\etbt$, we also have
$e^\req (s_{r-1}(\etbt);\theta_{r-1})=\etbt$, hence
\[
\delta_r(\etbt)
=
e^\req \!\big(s_{r-1}(\etbt);\theta_r\big)
-
 e^\req \!\big(s_{r-1}(\etbt);\theta_{r-1}\big).
\]
By (C3), the effort gap
$e^\req (s;\theta_r)- e^\req (s;\theta_{r-1})$ is nondecreasing in $S$; composing with the nondecreasing map
$\etbt\mapsto s_{r-1}(\etbt)$ implies $\delta_r(\etbt)$ is nondecreasing in $\etbt$.

Finally, the strict ``just-overtake'' effort $e^{\req}_r(\etbt)$ is obtained by requiring
$v(\theta_r,\etbt+e)>s_{r-1}(\etbt)$ instead of $\ge$; by continuity of $v(\theta_r,\cdot)$ this differs only by an
arbitrarily small $\varepsilon>0$, so $e^\req_r(\etbt)$ is also nondecreasing in $\etbt$.

Since $c$ is nondecreasing, 
$c(e^\req_r(\etbt))$ is nondecreasing as well.
\end{proof}

\paragraph{Proof for \Cref{prop:tbt-monotonic-incentive}}
\begin{proof}
\label{proof:prop:tbt-monotonic-incentive}
Given $0 \leq \etbt_1 < \etbt_2$, and assume the all-zero additional-effort profile $\textbf{e}=(0,\dots,0)$ is a PNE of the induced follower game under baseline $\etbt_1$, we know from \Cref{prop:zero-eq-tbt} that $\forall r\in \{2, \cdots n\}$, we have
\begin{align*}
    c (e^\req_r(\etbt_1)) \geq R_{r-1} - R_r.
\end{align*}
Then according to \Cref{lemma:tbt-increase-cost-of-climbing}, the cost of "just-overtake" effort $c(e^\req_r(\etbt))$ is monotonically increasing in $\etbt$, which means:
\begin{align*}
    c (e^\req_r(\etbt_2)) \geq  c (e^\req_r(\etbt_1)) \geq R_{r-1} - R_r,
\end{align*}
which implies that $\textbf{e} = (0, \cdots, 0)$ is also a PNE under $\etbt_2$.
\end{proof}

\paragraph{Proof for \Cref{cor:tbt-stabilizes-incentives}}
\begin{proof}
\label{proof:cor:tbt-stabilizes-incentives}
By definition of ${\etbt}^\star$, for any $\etbt\ge {\etbt}^\star$ we have
$c(e^\req_r(\etbt))\ge R_{r-1}-R_r$ for $r=\{2,\dots,n\}$, so the all-zero profile is a PNE by \Cref{prop:zero-eq-tbt}.
\end{proof}

\paragraph{Proof for \Cref{prop:powerlaw-linear-cost-stabilizing-tbt}}

\begin{proof}
Since $\sigma$ is strictly increasing, the catch-up condition
$v(\theta_r,\etbt+\delta)\ge v(\theta_{r-1},\etbt)$ is equivalent to
\[
\tilde v(\theta_r,\etbt+\delta)
\;\ge\;
\frac{v(\theta_{r-1},\etbt)-L(\theta_r)}{U(\theta_r)-L(\theta_r)}.
\]
Applying $\sigma^{-1}$ and the scaling law for $\tilde v(\theta_r,\cdot)$, this becomes
\[
\alpha(\theta_r)+\beta(\theta_r)\log(1+\etbt+\delta)
\;\ge\;
\sigma^{-1}\!\Bigl(\frac{v(\theta_{r-1},\etbt)-L(\theta_r)}{U(\theta_r)-L(\theta_r)}\Bigr).
\]
Solving with equality yields the minimal additional effort:
\[
1+\etbt+\delta
=
\exp\!\Bigl(\tfrac{1}{\beta(\theta_r)}\Bigl[
\sigma^{-1}\!\bigl(\tfrac{v(\theta_{r-1},\etbt)-L(\theta_r)}{U(\theta_r)-L(\theta_r)}\bigr)
-\alpha(\theta_r)\Bigr]\Bigr).
\]
Define $e_r^{\req}(\etbt)$ to be this minimal $\delta$, then subtracting $(1+\etbt)$ gives
\[
e_r^{\req}(\etbt)
=
\exp\!\Bigl(\tfrac{1}{\beta(\theta_r)}\Bigl[
\sigma^{-1}\!\bigl(\tfrac{v(\theta_{r-1},\etbt)-L(\theta_r)}{U(\theta_r)-L(\theta_r)}\bigr)
-\alpha(\theta_r)\Bigr]\Bigr)
-(1+\etbt).
\]
Using the scaling law for $\theta_{r-1}$,
\[
\sigma^{-1}\!\bigl(\tilde v(\theta_{r-1},\etbt)\bigr)
=
\alpha(\theta_{r-1})+\beta(\theta_{r-1})\log(1+\etbt),
\]
and the definition of $\gamma_r=\beta(\theta_r)/\beta(\theta_{r-1})$, we can rewrite
the leading term as
\[
\exp\!\Bigl(\frac{\bar\alpha_r(\etbt)}{\beta(\theta_r)}\Bigr)\,(1+\etbt)^{1/\gamma_r},
\]
where $\bar\alpha_r(\etbt)=\sigma^{-1}(\tilde v(\theta_{r-1},\etbt))-\alpha(\theta_r)$,
yielding the stated expression.

With costs $c(e)\ge \kappa e$, any deviation that uses additional effort $\delta$
incurs cost at least $\kappa\delta$, so an adjacent-overtake deviation is unprofitable
whenever $e_r^{\req}(\etbt)\ge \rho_r$.

Finally, in the regime where the leading term in $e_r^{\req}(\etbt)$ dominates
$(1+\etbt)$, the condition $e_r^{\req}(\etbt)\ge \rho_r$ is well-approximated by
\[
\exp\!\Bigl(\frac{\bar\alpha_r(\etbt)}{\beta(\theta_r)}\Bigr)\,(1+\etbt)^{1/\gamma_r}
\;\gtrsim\; \rho_r.
\]
Treating $\bar\alpha_r(\etbt)$ as approximately constant over the relevant range
(or evaluating it at $\etbt=0$ to obtain a conservative rule of thumb) gives
\[
1+\etbt \;\gtrsim\;
\left(\frac{\rho_r}{\exp(\bar\alpha_r(0)/\beta(\theta_r))}\right)^{\gamma_r}
=
\left(\frac{\rho_r}{e_r^{\req}(0)+1}\right)^{\gamma_r},
\]
which implies the stated scaling. Since $\lambda_r=\rho_r/e_r^{\req}(0)$, we also have
${\etbt}_r^\star=O(\lambda_r^{\gamma_r})$ up to a constant factor depending only on
$e_r^{\req}(0)$.
\end{proof}

\paragraph{Estimating $e_r^{\req}(0)$, $\gamma_r$, and $\lambda_r$ from the post-training trajectories.}
For each model $i$, we fit the generalized scaling form
$\sigma^{-1}(\tilde v_i(e))=\alpha_i+\beta_i\log(1+e)$ by linear regression of
$\sigma^{-1}(\tilde v_i)$ on $\log(1+e)$ over the observed training range,
yielding $(\hat\alpha_i,\hat\beta_i)$. For an adjacent pair $(r-1,r)$ we set
\[
\hat\gamma_r := \frac{\hat\beta_r}{\hat\beta_{r-1}}.
\]
The baseline catch-up difficulty $e_r^{\req}(0)$ is defined as the minimal additional
effort at $\etbt=0$ required for model $r$ to match model $r-1$; under the fitted scaling law,
\[
\widehat{e_r^{\req}(0)}
=
\exp\!\Bigl(\tfrac{\hat\alpha_{r-1}-\hat\alpha_r}{\hat\beta_r}\Bigr)-1.
\]
Given a reward gap (converted to effort units) $\rho_r=(R_{r-1}-R_r)/\kappa$, we form
the effective incentive $\hat\lambda_r:=\rho_r/\widehat{e_r^{\req}(0)}$.

As a nonparametric check, we also estimate a local slope ratio at a common baseline
$e=\etbt$ using finite differences on the fitted trajectories:
\[
\widehat{s}_i(\etbt)\approx \frac{\hat v_i(\etbt+h)-\hat v_i(\etbt)}{h},
\qquad
\widehat{\gamma}^{\,\mathrm{slope}}_r(\etbt):=\frac{\widehat{s}_r(\etbt)}{\widehat{s}_{r-1}(\etbt)}.
\]

\newpage
\section{Additional Empirical Results}
\label{app:empirical}
\begin{figure}
    \centering
    \includegraphics[width=0.98\linewidth]{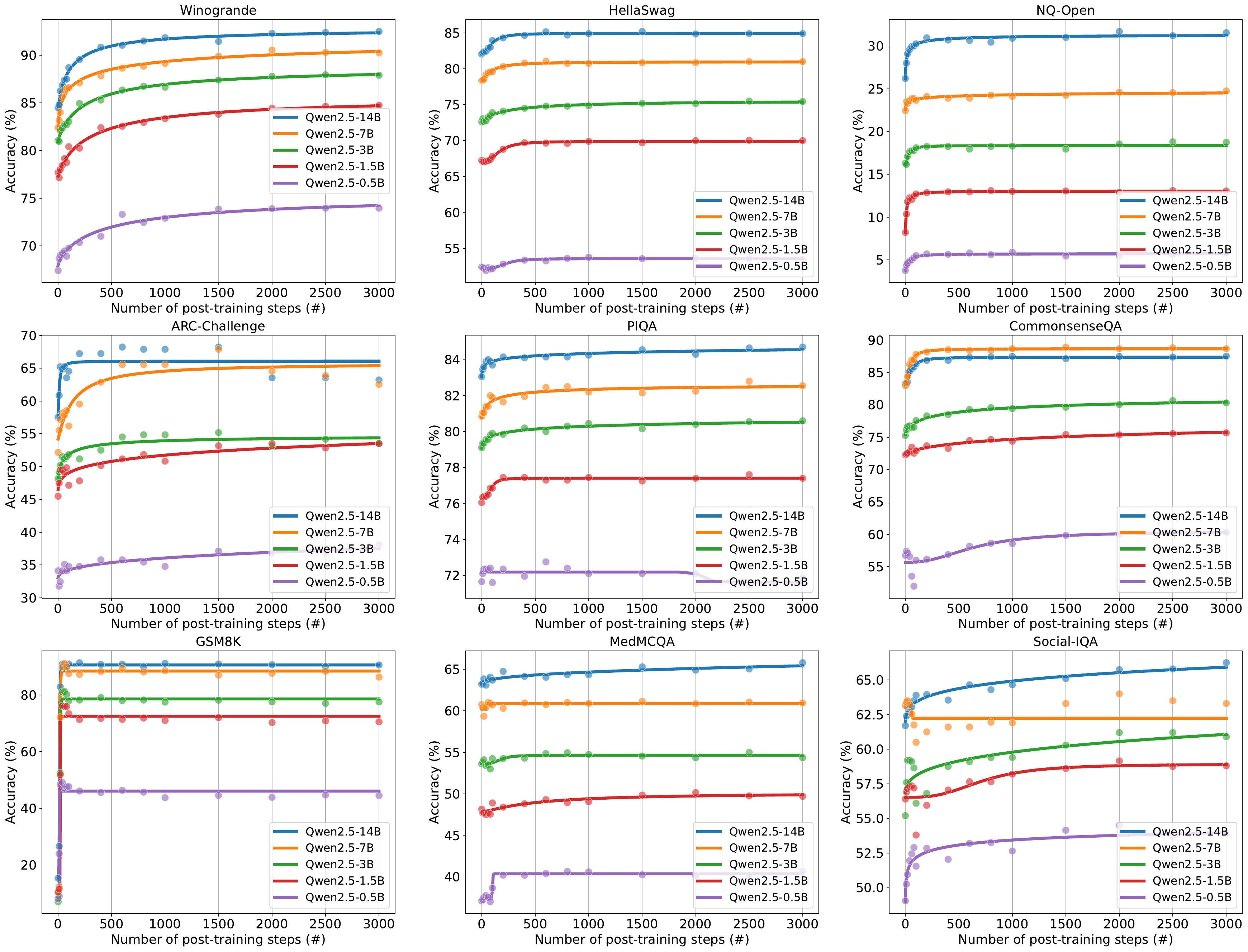}
    \caption{
    Continued post-training trajectories of \textit{Qwen} models of different sizes on nine benchmarks. Here, we use model size as a proxy for the model’s latent capability $\theta$. 
    The $x$-axis denotes the number of post-training steps, reflecting post-training effort $e$. 
    The $y$-axis denotes accuracy on the validation set, i.e., $v(\theta, e)$. 
    For each model, we fit a curve following Equation~\eqref{eq:log-score}. 
    }
    \label{fig:traj_full}
\end{figure}

\paragraph{Setting}
We conduct our emprical study on nine benchmarks, \texttt{Winogrande}~\citep{winogrande}, \texttt{HellaSwag}~\citep{hellaswag}, \texttt{NQ-Open}~\citep{nq_open}, \texttt{ARC-Challenge}~\citep{arc}, \texttt{Piqa}~\citep{piqa}, \texttt{CommonsenseQA}~\citep{commonsenseqa}, \texttt{Gsm8k}~\citep{gsm8k}, \texttt{MedMcQA}~\citep{medmcqa}, and \texttt{Social-IQA}~\citep{social_iqa}.
Each model is trained with LoRA~\citep{Hu2021LoRALA,peft} (rank 8, $\alpha$=32) and AdamW optimizer~\citep{Loshchilov2017DecoupledWD} (weight decay 0.01 and learning rate 5e-5).
We use a batch size of 8, so each training step corresponds to eight training data points.
For those benchmarks without a validation split, we randomly allocate 20\% of the training data as the validation set.

\paragraph{Results}
In Figure~\ref{fig:traj_full}, we present the results of controlled post-training experiments for all nine benchmarks.
The assumption~\ref{ass:post-effort-score} still largely holds in most benchmarks, despite a few anomalies.
In addition, on benchmarks like \texttt{HellaSwag}, continued post-training cannot make a smaller model catch up to a larger one. 
This is not surprising, since model sizes (the proxy for $\theta$) differ significantly in our controlled setting. 
As a result, there is no incentive to conduct strategic post-training in this scenario. 
We argue, however, that in real-world settings with models from diverse sources, $\theta$ will not differ so dramatically.

%% file: references.bib
@article{dominguez2024training,
  title={Training on the test task confounds evaluation and emergence},
  author={Dominguez-Olmedo, Ricardo and Dorner, Florian E and Hardt, Moritz},
  journal={arXiv preprint arXiv:2407.07890},
  year={2024}
}

@inproceedings{liu2022strategic,
  title={Strategic ranking},
  author={Liu, Lydia T and Garg, Nikhil and Borgs, Christian},
  booktitle={International Conference on Artificial Intelligence and Statistics},
  pages={2489--2518},
  year={2022},
  organization={PMLR}
}

@article{zhang2025train,
  title={Train-before-Test Harmonizes Language Model Rankings},
  author={Zhang, Guanhua and Dominguez-Olmedo, Ricardo and Hardt, Moritz},
  journal={arXiv preprint arXiv:2507.05195},
  year={2025}
}

@article{ruan2024observational,
  title={Observational scaling laws and the predictability of langauge model performance},
  author={Ruan, Yangjun and Maddison, Chris J and Hashimoto, Tatsunori B},
  journal={Advances in Neural Information Processing Systems},
  volume={37},
  pages={15841--15892},
  year={2024}
}

@inproceedings{hardt2016strategic,
  title={Strategic classification},
  author={Hardt, Moritz and Megiddo, Nimrod and Papadimitriou, Christos and Wootters, Mary},
  booktitle={Proceedings of the 2016 ACM conference on innovations in theoretical computer science},
  pages={111--122},
  year={2016}
}

@article{kendall1938rank,
 ISSN = {00063444},
 URL = {http://www.jstor.org/stable/2332226},
 author = {M. G. Kendall},
 journal = {Biometrika},
 number = {1/2},
 pages = {81--93},
 publisher = {[Oxford University Press, Biometrika Trust]},
 title = {A New Measure of Rank Correlation},
 urldate = {2025-09-25},
 volume = {30},
 year = {1938}
}

@misc{datologyai2024textcuration,
  author       = {{DatologyAI}},
  title        = {Technical Deep-Dive: Curating Our Way to a State-of-the-Art Text Dataset},
  year         = {2024},
  month        = nov,
  howpublished = {{https://www.datologyai.com/blog/technical-deep-dive-curating-our-way-to-a-state-of-the-art-text-dataset}}
}

@article{Yang2024Qwen25TR,
  title={Qwen2.5 Technical Report},
  author={Qwen An Yang and Baosong Yang and Beichen Zhang and Binyuan Hui and Bo Zheng and Bowen Yu and Chengyuan Li and Dayiheng Liu and Fei Huang and Guanting Dong and Haoran Wei and Huan Lin and Jian Yang and Jianhong Tu and Jianwei Zhang and Jianxin Yang and Jiaxin Yang and Jingren Zhou and Junyang Lin and Kai Dang and Keming Lu and Keqin Bao and Kexin Yang and Le Yu and Mei Li and Mingfeng Xue and Pei Zhang and Qin Zhu and Rui Men and Runji Lin and Tianhao Li and Tingyu Xia and Xingzhang Ren and Xuancheng Ren and Yang Fan and Yang Su and Yi-Chao Zhang and Yunyang Wan and Yuqi Liu and Zeyu Cui and Zhenru Zhang and Zihan Qiu and Shanghaoran Quan and Zekun Wang},
  journal={ArXiv},
  year={2024},
  volume={abs/2412.15115},
  url={https://api.semanticscholar.org/CorpusID:274859421}
}

@article{singh2025leaderboard,
  title={The leaderboard illusion},
  author={Singh, Shivalika and Nan, Yiyang and Wang, Alex and D'Souza, Daniel and Kapoor, Sayash and {\"U}st{\"u}n, Ahmet and Koyejo, Sanmi and Deng, Yuntian and Longpre, Shayne and Smith, Noah A and others},
  journal={arXiv preprint arXiv:2504.20879},
  year={2025}
}

@article{phan2025humanity,
  title={Humanity's last exam},
  author={Phan, Long and Gatti, Alice and Han, Ziwen and Li, Nathaniel and Hu, Josephina and Zhang, Hugh and Zhang, Chen Bo Calvin and Shaaban, Mohamed and Ling, John and Shi, Sean and others},
  journal={arXiv preprint arXiv:2501.14249},
  year={2025}
}

@article{hendrycks2020measuring,
  title={Measuring massive multitask language understanding},
  author={Hendrycks, Dan and Burns, Collin and Basart, Steven and Zou, Andy and Mazeika, Mantas and Song, Dawn and Steinhardt, Jacob},
  journal={arXiv preprint arXiv:2009.03300},
  year={2020}
}

@article{zellers2019hellaswag,
  title={Hellaswag: Can a machine really finish your sentence?},
  author={Zellers, Rowan and Holtzman, Ari and Bisk, Yonatan and Farhadi, Ali and Choi, Yejin},
  journal={arXiv preprint arXiv:1905.07830},
  year={2019}
}

@inproceedings{Loshchilov2017DecoupledWD,
  title={Decoupled Weight Decay Regularization},
  author={Ilya Loshchilov and Frank Hutter},
  booktitle={International Conference on Learning Representations},
  year={2017},
  url={https://api.semanticscholar.org/CorpusID:53592270}
}

@Misc{peft,
  title =        {{PEFT}: State-of-the-art Parameter-Efficient Fine-Tuning methods},
  author =       {Sourab Mangrulkar and Sylvain Gugger and Lysandre Debut and Younes Belkada and Sayak Paul and Benjamin Bossan and Marian Tietz},
  howpublished = {\url{https://github.com/huggingface/peft}},
  year =         {2022}
}

@article{Hu2021LoRALA,
  title={LoRA: Low-Rank Adaptation of Large Language Models},
  author={J. Edward Hu and Yelong Shen and Phillip Wallis and Zeyuan Allen-Zhu and Yuanzhi Li and Shean Wang and Weizhu Chen},
  journal={ArXiv},
  year={2021},
  volume={abs/2106.09685},
  url={https://api.semanticscholar.org/CorpusID:235458009}
}

@misc{finnveden2020extrapolating,
  author       = {Lukas Finnveden},
  title        = {Extrapolating GPT-N Performance},
  year         = {2020},
  url          = {https://www.lesswrong.com/posts/k2SNji3jXaLGhBeYP/extrapolating-gpt-n-performance},
}

@article{owen2024predictable,
  author       = {David Owen},
  title        = {How Predictable is Language Model Benchmark Performance?},
  journal      = {arXiv preprint},
  volume       = {arXiv:2401.04757},
  year         = {2024},
  url          = {https://arxiv.org/abs/2401.04757},
}

@book{topkis1998supermodularity,
 ISBN = {9780691032443},
 URL = {http://www.jstor.org/stable/j.ctt7s83q},
  author={Donald M. Topkis},
 publisher = {Princeton University Press},
 title = {Supermodularity and Complementarity},
 urldate = {2026-01-18},
 year = {1998}
}

@article{nash1950equilibirum,
  title={Equilibrium points in n-person games},
  author={Nash, John F},
  journal={Proceedings of the National Academy of Sciences},
  volume={36},
  number={1},
  pages={48--49},
  year={1950},
  publisher={National Acad Sciences}
}

@article{von2010leadership,
  title={Leadership games with convex strategy sets},
  author={Von Stengel, Bernhard and Zamir, Shmuel},
  journal={Games and Economic Behavior},
  volume={69},
  number={2},
  pages={446--457},
  year={2010},
  publisher={Elsevier}
}

@Inbook{marchesi2021leadership,
author="Marchesi, Alberto",
title="Leadership Games: Multiple Followers, Multiple Leaders, and Perfection",
bookTitle="Special Topics in Information Technology",
year="2021",
publisher="Springer International Publishing",
address="Cham",
pages="107--118",
isbn="978-3-030-62476-7",
doi="10.1007/978-3-030-62476-7_10",
url="https://doi.org/10.1007/978-3-030-62476-7_10"
}

@article{wei2021finetuned,
  title={Finetuned language models are zero-shot learners},
  author={Wei, Jason and Bosma, Maarten and Zhao, Vincent Y and Guu, Kelvin and Yu, Adams Wei and Lester, Brian and Du, Nan and Dai, Andrew M and Le, Quoc V},
  journal={arXiv preprint arXiv:2109.01652},
  year={2021}
}

@article{chen2023learning,
  title={Learning to incentivize improvements from strategic agents},
  author={Chen, Yatong and Wang, Jialu and Liu, Yang},
  journal={Transactions on Machine Learning Research},
  year={2023}
}

@article{Donoho2023DataSA,
  title={Data science at the singularity},
  author={Donoho, David},
  journal={Harvard Data Science Review},
  volume={6},
  number={1},
  year={2024},
  publisher={MIT Press}
}

@inproceedings{Duda1974PatternCA,
  title={Pattern classification and scene analysis},
  author={Richard O. Duda and Peter E. Hart},
  booktitle={A Wiley-Interscience publication},
  year={1974},
  url={https://api.semanticscholar.org/CorpusID:12946615}
}

@misc{hardt2025emerging,
  author = {Moritz Hardt},
  title = {The Emerging Science of Machine Learning Benchmarks},
  year = {2025},
  howpublished = {Online at \url{https://mlbenchmarks.org}},
  note = {Manuscript}
}

@article{Russakovsky2014ImageNetLS,
  title={ImageNet Large Scale Visual Recognition Challenge},
  author={Olga Russakovsky and Jia Deng and Hao Su and Jonathan Krause and Sanjeev Satheesh and Sean Ma and Zhiheng Huang and Andrej Karpathy and Aditya Khosla and Michael S. Bernstein and Alexander C. Berg and Li Fei-Fei},
  journal={International Journal of Computer Vision},
  year={2014},
  volume={115},
  pages={211 - 252},
  url={https://api.semanticscholar.org/CorpusID:2930547}
}

@article{He2015DeepRL,
  title={Deep Residual Learning for Image Recognition},
  author={Kaiming He and X. Zhang and Shaoqing Ren and Jian Sun},
  journal={2016 IEEE Conference on Computer Vision and Pattern Recognition (CVPR)},
  year={2015},
  pages={770-778},
  url={https://api.semanticscholar.org/CorpusID:206594692}
}

@article{Krizhevsky2012ImageNetCW,
  title={ImageNet classification with deep convolutional neural networks},
  author={Alex Krizhevsky and Ilya Sutskever and Geoffrey E. Hinton},
  journal={Communications of the ACM},
  year={2012},
  volume={60},
  pages={84 - 90},
  url={https://api.semanticscholar.org/CorpusID:195908774}
}

@inproceedings{Liao2021AreWL,
  title={Are We Learning Yet? A Meta Review of Evaluation Failures Across Machine Learning},
  author={Thomas Liao},
  booktitle={NeurIPS Datasets and Benchmarks},
  year={2021},
  url={https://api.semanticscholar.org/CorpusID:244907059}
}

@book{hardtrecht2022patterns,
  author = {Moritz Hardt and Benjamin Recht},
  title = {Patterns, predictions, and actions: Foundations of machine learning},
  year = {2022},
  publisher = {Princeton University Press}
}

@book{Goodfellow-et-al-2016,
    title={Deep Learning},
    author={Ian Goodfellow and Yoshua Bengio and Aaron Courville},
    publisher={MIT Press},
    note={\url{http://www.deeplearningbook.org}},
    year={2016}
}

@article{Dwork2015GeneralizationIA,
  title={Generalization in Adaptive Data Analysis and Holdout Reuse},
  author={Cynthia Dwork and Vitaly Feldman and Moritz Hardt and Toniann Pitassi and Omer Reingold and Aaron Roth},
  journal={ArXiv},
  year={2015},
  volume={abs/1506.02629},
  url={https://api.semanticscholar.org/CorpusID:14762349}
}

@article{Dwork2015TheRH,
  title={The reusable holdout: Preserving validity in adaptive data analysis},
  author={Cynthia Dwork and Vitaly Feldman and Moritz Hardt and Toniann Pitassi and Omer Reingold and Aaron Roth},
  journal={Science},
  year={2015},
  volume={349},
  pages={636 - 638},
  url={https://api.semanticscholar.org/CorpusID:15569600}
}

@inproceedings{Blum2015TheLA,
  title={The Ladder: A Reliable Leaderboard for Machine Learning Competitions},
  author={Avrim Blum and Moritz Hardt},
  booktitle={International Conference on Machine Learning},
  year={2015},
  url={https://api.semanticscholar.org/CorpusID:1493191}
}

@article{Mania2019ModelSM,
  title={Model Similarity Mitigates Test Set Overuse},
  author={Horia Mania and John Miller and Ludwig Schmidt and Moritz Hardt and Benjamin Recht},
  journal={ArXiv},
  year={2019},
  volume={abs/1905.12580},
  url={https://api.semanticscholar.org/CorpusID:168169971}
}

@inproceedings{Yadav2019ColdCT,
  title={Cold Case: The Lost MNIST Digits},
  author={Chhavi Yadav and L{\'e}on Bottou},
  booktitle={Neural Information Processing Systems},
  year={2019},
  url={https://api.semanticscholar.org/CorpusID:166227957}
}

@inproceedings{Recht2019DoIC,
  title={Do ImageNet Classifiers Generalize to ImageNet?},
  author={Benjamin Recht and Rebecca Roelofs and Ludwig Schmidt and Vaishaal Shankar},
  booktitle={International Conference on Machine Learning},
  year={2019},
  url={https://api.semanticscholar.org/CorpusID:67855879}
}

@article{Miller2020TheEO,
  title={The Effect of Natural Distribution Shift on Question Answering Models},
  author={John Miller and Karl Krauth and Benjamin Recht and Ludwig Schmidt},
  journal={ArXiv},
  year={2020},
  volume={abs/2004.14444},
  url={https://api.semanticscholar.org/CorpusID:216867120}
}

@article{Kornblith2018DoBI,
  title={Do Better ImageNet Models Transfer Better?},
  author={Simon Kornblith and Jonathon Shlens and Quoc V. Le},
  journal={2019 IEEE/CVF Conference on Computer Vision and Pattern Recognition (CVPR)},
  year={2018},
  pages={2656-2666},
  url={https://api.semanticscholar.org/CorpusID:43928547}
}

@article{Deng2009ImageNetAL,
  title={ImageNet: A large-scale hierarchical image database},
  author={Jia Deng and Wei Dong and Richard Socher and Li-Jia Li and K. Li and Li Fei-Fei},
  journal={2009 IEEE Conference on Computer Vision and Pattern Recognition},
  year={2009},
  pages={248-255},
  url={https://api.semanticscholar.org/CorpusID:57246310}
}

@article{Salaudeen2024ImageNotAC,
  title={ImageNot: A contrast with ImageNet preserves model rankings},
  author={Olawale Salaudeen and Moritz Hardt},
  journal={ArXiv},
  year={2024},
  volume={abs/2404.02112},
  url={https://api.semanticscholar.org/CorpusID:268857319}
}

@article{Salaudeen2025MeasurementTM,
  title={Measurement to Meaning: A Validity-Centered Framework for AI Evaluation},
  author={Olawale Salaudeen and Anka Reuel and Ahmed M. Ahmed and Suhana Bedi and Zachary Robertson and Sudharsan Sundar and Ben Domingue and Angelina Wang and Oluwasanmi Koyejo},
  journal={ArXiv},
  year={2025},
  volume={abs/2505.10573},
  url={https://api.semanticscholar.org/CorpusID:278715024}
}

@article{Zhang2024InherentTB,
  title={Inherent Trade-Offs between Diversity and Stability in Multi-Task Benchmarks},
  author={Guanhua Zhang and Moritz Hardt},
  journal={ArXiv},
  year={2024},
  volume={abs/2405.01719},
  url={https://api.semanticscholar.org/CorpusID:269587753}
}

@article{Liang2023HolisticEO,
  title={Holistic Evaluation of Language Models},
  author={Percy Liang and Rishi Bommasani and Tony Lee and Dimitris Tsipras and Dilara Soylu and Michihiro Yasunaga and Yian Zhang and Deepak Narayanan and Yuhuai Wu and Ananya Kumar and Benjamin Newman and Binhang Yuan and Bobby Yan and Ce Zhang and Christian Cosgrove and Christopher D. Manning and Christopher R'e and Diana Acosta-Navas and Drew A. Hudson and E. Zelikman and Esin Durmus and Faisal Ladhak and Frieda Rong and Hongyu Ren and Huaxiu Yao and Jue Wang and Keshav Santhanam and Laurel J. Orr and Lucia Zheng and Mert Yuksekgonul and Mirac Suzgun and Nathan S. Kim and Neel Guha and Niladri S. Chatterji and O. Khattab and Peter Henderson and Qian Huang and Ryan Chi and Sang Michael Xie and Shibani Santurkar and Surya Ganguli and Tatsunori Hashimoto and Thomas F. Icard and Tianyi Zhang and Vishrav Chaudhary and William Wang and Xuechen Li and Yifan Mai and Yuhui Zhang and Yuta Koreeda},
  journal={Annals of the New York Academy of Sciences},
  year={2023},
  volume={1525},
  pages={140 - 146},
  url={https://api.semanticscholar.org/CorpusID:253553585}
}

@misc{open-llm-leaderboard-v2,
  author = {Clémentine Fourrier and Nathan Habib and Alina Lozovskaya and Konrad Szafer and Thomas Wolf},
  title = {Open LLM Leaderboard v2},
  year = {2024},
  publisher = {Hugging Face},
  howpublished = "\url{https://huggingface.co/spaces/open-llm-leaderboard/open_llm_leaderboard}",
}

@article{Touvron2023Llama2O,
  title={Llama 2: Open Foundation and Fine-Tuned Chat Models},
  author={Hugo Touvron and Louis Martin and Kevin R. Stone and Peter Albert and Amjad Almahairi and Yasmine Babaei and Niko-lay Bashlykov and Soumya Batra and Prajjwal Bhargava and Shruti Bhosale and Daniel M. Bikel and Lukas Blecher and Cristian Canton Ferrer and Moya Chen and Guillem Cucurull and David Esiobu and Jude Fernandes and Jeremy Fu and Wenyin Fu and Brian Fuller and Cynthia Gao and Vedanuj Goswami and Naman Goyal and Anthony S. Hartshorn and Saghar Hosseini and Rui Hou and Hakan Inan and Marcin Kardas and Viktor Kerkez and Madian Khabsa and Isabel M. Kloumann and Artem Korenev and Punit Singh Koura and Marie-Anne Lachaux and Thibaut Lavril and Jenya Lee and Diana Liskovich and Yinghai Lu and Yuning Mao and Xavier Martinet and Todor Mihaylov and Pushkar Mishra and Igor Molybog and Yixin Nie and Andrew Poulton and Jeremy Reizenstein and Rashi Rungta and Kalyan Saladi and Alan Schelten and Ruan Silva and Eric Michael Smith and R. Subramanian and Xia Tan and Binh Tang and Ross Taylor and Adina Williams and Jian Xiang Kuan and Puxin Xu and Zhengxu Yan and Iliyan Zarov and Yuchen Zhang and Angela Fan and Melissa Hall Melanie Kambadur and Sharan Narang and Aur'elien Rodriguez and Robert Stojnic and Sergey Edunov and Thomas Scialom},
  journal={ArXiv},
  year={2023},
  volume={abs/2307.09288},
  url={https://api.semanticscholar.org/CorpusID:259950998}
}

@article{Zhou2023DontMY,
  title={Don't Make Your LLM an Evaluation Benchmark Cheater},
  author={Kun Zhou and Yutao Zhu and Zhipeng Chen and Wentong Chen and Wayne Xin Zhao and Xu Chen and Yankai Lin and Ji-Rong Wen and Jiawei Han},
  journal={ArXiv},
  year={2023},
  volume={abs/2311.01964},
  url={https://api.semanticscholar.org/CorpusID:265019021}
}

@article{Schaeffer2023AreEA,
  title={Are Emergent Abilities of Large Language Models a Mirage?},
  author={Rylan Schaeffer and Brando Miranda and Oluwasanmi Koyejo},
  journal={ArXiv},
  year={2023},
  volume={abs/2304.15004},
  url={https://api.semanticscholar.org/CorpusID:258418299}
}

@article{Manheim2018CategorizingVO,
  title={Categorizing Variants of Goodhart's Law},
  author={David Manheim and Scott Garrabrant},
  journal={ArXiv},
  year={2018},
  volume={abs/1803.04585},
  url={https://api.semanticscholar.org/CorpusID:4715794}
}

@article{Bordt2024HowMC,
  title={How much can we forget about Data Contamination?},
  author={Sebastian Bordt and Suraj Srinivas and Valentyn Boreiko and Ulrike von Luxburg},
  journal={ArXiv},
  year={2024},
  volume={abs/2410.03249},
  url={https://api.semanticscholar.org/CorpusID:273163321}
}

@article{Yang2023RethinkingBA,
  title={Rethinking Benchmark and Contamination for Language Models with Rephrased Samples},
  author={Shuo Yang and Wei-Lin Chiang and Lianmin Zheng and Joseph Gonzalez and Ion Stoica},
  journal={ArXiv},
  year={2023},
  volume={abs/2311.04850},
  url={https://api.semanticscholar.org/CorpusID:265050721}
}

@article{Jiang2024InvestigatingDC,
  title={Investigating Data Contamination for Pre-training Language Models},
  author={Minhao Jiang and Ken Ziyu Liu and Ming Zhong and Rylan Schaeffer and Siru Ouyang and Jiawei Han and Sanmi Koyejo},
  journal={ArXiv},
  year={2024},
  volume={abs/2401.06059},
  url={https://api.semanticscholar.org/CorpusID:266933004}
}

@article{Li2024DataCompLMIS,
  title={DataComp-LM: In search of the next generation of training sets for language models},
  author={Jeffrey Li and Alex Fang and Georgios Smyrnis and Maor Ivgi and Matt Jordan and Samir Yitzhak Gadre and Hritik Bansal and Etash Kumar Guha and Sedrick Scott Keh and Kushal Arora and Saurabh Garg and Rui Xin and Niklas Muennighoff and Reinhard Heckel and Jean-Pierre Mercat and Mayee Chen and Suchin Gururangan and Mitchell Wortsman and Alon Albalak and Yonatan Bitton and Marianna Nezhurina and Amro Abbas and Cheng-Yu Hsieh and Dhruba Ghosh and Josh Gardner and Maciej Kilian and Hanlin Zhang and Rulin Shao and Sarah Pratt and Sunny Sanyal and Gabriel Ilharco and Giannis Daras and Kalyani Marathe and Aaron Gokaslan and Jieyu Zhang and Khy-athi Chandu and Thao Nguyen and Igor Vasiljevic and Sham M. Kakade and Shuran Song and Sujay Sanghavi and Fartash Faghri and Sewoong Oh and Luke S. Zettlemoyer and Kyle Lo and Alaaeldin El-Nouby and Hadi Pouransari and Alexander Toshev and Stephanie Wang and Dirk Groeneveld and Luca Soldani and Pang Wei Koh and Jenia Jitsev and Thomas Kollar and Alexandros G. Dimakis and Yair Carmon and Achal Dave and Ludwig Schmidt and Vaishaal Shankar},
  journal={ArXiv},
  year={2024},
  volume={abs/2406.11794},
  url={https://api.semanticscholar.org/CorpusID:270560330}
}

@article{Guha2025OpenThoughtsDR,
  title={OpenThoughts: Data Recipes for Reasoning Models},
  author={Etash Kumar Guha and Ryan Marten and Sedrick Scott Keh and Negin Raoof and Georgios Smyrnis and Hritik Bansal and Marianna Nezhurina and Jean-Pierre Mercat and Trung Vu and Zayne Sprague and Ashima Suvarna and Ben Feuer and Liangyu Chen and Zaid Khan and Eric Frankel and Sachin Grover and Caroline Choi and Niklas Muennighoff and Shiye Su and Wanjia Zhao and John Yang and Shreyas Pimpalgaonkar and Kartik Sharma and Charlie Cheng-Jie Ji and Yichuan Deng and Sarah Pratt and Vivek Ramanujan and Jon Saad-Falcon and Jeffrey Li and Achal Dave and Alon Albalak and Kushal Arora and Blake Wulfe and Chinmay Hegde and Greg Durrett and Sewoong Oh and Mohit Bansal and Saadia Gabriel and Aditya Grover and Kai-Wei Chang and Vaishaal Shankar and Aaron Gokaslan and Mike A. Merrill and Tatsunori Hashimoto and Yejin Choi and Jenia Jitsev and Reinhard Heckel and Maheswaran Sathiamoorthy and Alexandros G. Dimakis and Ludwig Schmidt},
  journal={ArXiv},
  year={2025},
  volume={abs/2506.04178},
  url={https://api.semanticscholar.org/CorpusID:279154475}
}

@article{einav2025market,
  title         = {A Market for Accuracy: Classification under Competition},
  author        = {Einav, Ohad and Rosenfeld, Nir},
  journal       = {arXiv preprint arXiv:2502.18052},
  year          = {2025},
  eprint        = {2502.18052},
  archivePrefix = {arXiv},
  primaryClass  = {cs.LG}
}

@inproceedings{rosenfeld2024strategicml,
  title     = {Strategic ML: How to Learn With Data That ``Behaves''},
  author    = {Rosenfeld, Nir},
  booktitle = {Proceedings of the 17th ACM International Conference on Web Search and Data Mining (WSDM)},
  pages     = {1128--1131},
  year      = {2024},
  publisher = {ACM}
}

@inproceedings{bruckner2011stackelberg,
  title={Stackelberg games for adversarial prediction problems},
  author={Br{\"u}ckner, Michael and Scheffer, Tobias},
  booktitle={Proceedings of the 17th ACM SIGKDD international conference on Knowledge discovery and data mining},
  pages={547--555},
  year={2011}
}

@article{kleinberg2020classifiers,
  title={How Do Classifiers Induce Agents to Invest Effort Strategically?},
  author={Kleinberg, Jon and Raghavan, Manish},
  journal={ACM Transactions on Economics and Computation (TEAC)},
  volume={8},
  number={4},
  pages={1--23},
  year={2020},
  publisher={ACM New York, NY, USA}
}

@inproceedings{miller2020strategic,
  title={Strategic classification is causal modeling in disguise},
  author={Miller, John and Milli, Smitha and Hardt, Moritz},
  booktitle={International Conference on Machine Learning},
  pages={6917--6926},
  year={2020},
  organization={PMLR}
}

@inproceedings{sommer2025learning,
  title        = {Learning Classifiers That Induce Markets},
  author       = {Sommer, Yonatan and Hikri, Ivri and Amit, Lotan and Rosenfeld, Nir},
  booktitle    = {Proceedings of the 42nd International Conference on Machine Learning (ICML 2025)},
  year         = {2025},
  arxiv        = {2502.20012},
}

@article{Patwardhan2025GDPvalEA,
  title={GDPval: Evaluating AI Model Performance on Real-World Economically Valuable Tasks},
  author={Tejal Patwardhan and Rachel Dias and Elizabeth Proehl and Grace Kim and Michele Wang and Olivia Watkins and Sim'on Posada Fishman and Marwan Aljubeh and Phoebe Thacker and Laurance Fauconnet and Natalie S. Kim and Patrick Chao and Samuel Miserendino and Gildas Chabot and David Li and Michael Sharman and Alexandra Barr and Amelia Glaese and Jerry Tworek},
  journal={ArXiv},
  year={2025},
  volume={abs/2510.04374},
  url={https://api.semanticscholar.org/CorpusID:281843768}
}

@article{Jain2024LiveCodeBenchHA,
  title={LiveCodeBench: Holistic and Contamination Free Evaluation of Large Language Models for Code},
  author={Naman Jain and King Han and Alex Gu and Wen-Ding Li and Fanjia Yan and Tianjun Zhang and Sida Wang and Armando Solar-Lezama and Koushik Sen and Ion Stoica},
  journal={ArXiv},
  year={2024},
  volume={abs/2403.07974},
  url={https://api.semanticscholar.org/CorpusID:268379413}
}

@article{Albalak2024ASO,
  title={A Survey on Data Selection for Language Models},
  author={Alon Albalak and Yanai Elazar and Sang Michael Xie and Shayne Longpre and Nathan Lambert and Xinyi Wang and Niklas Muennighoff and Bairu Hou and Liangming Pan and Haewon Jeong and Colin Raffel and Shiyu Chang and Tatsunori Hashimoto and William Yang Wang},
  journal={ArXiv},
  year={2024},
  volume={abs/2402.16827},
  url={https://api.semanticscholar.org/CorpusID:268032975}
}

@article{Raffel2019ExploringTL,
  title={Exploring the Limits of Transfer Learning with a Unified Text-to-Text Transformer},
  author={Colin Raffel and Noam Shazeer and Adam Roberts and Katherine Lee and Sharan Narang and Michael Matena and Yanqi Zhou and Wei Li and Peter J. Liu},
  journal={J. Mach. Learn. Res.},
  year={2019},
  volume={21},
  pages={140:1-140:67},
  url={https://api.semanticscholar.org/CorpusID:204838007}
}

@article{Zhou2023InstructionFollowingEF,
  title={Instruction-Following Evaluation for Large Language Models},
  author={Jeffrey Zhou and Tianjian Lu and Swaroop Mishra and Siddhartha Brahma and Sujoy Basu and Yi Luan and Denny Zhou and Le Hou},
  journal={ArXiv},
  year={2023},
  volume={abs/2311.07911},
  url={https://api.semanticscholar.org/CorpusID:265157752}
}

@inproceedings{TjongKimSang2003IntroductionTT,
  title={Introduction to the CoNLL-2003 Shared Task: Language-Independent Named Entity Recognition},
  author={E. Tjong Kim Sang and Fien De Meulder},
  booktitle={Conference on Computational Natural Language Learning},
  year={2003},
  url={https://api.semanticscholar.org/CorpusID:2470716}
}

@misc{lecun-mnisthandwrittendigit-2010,
  author = {LeCun, Yann and Cortes, Corinna},
  howpublished = {http://yann.lecun.com/exdb/mnist/},
  lastchecked = {2016-01-14 14:24:11},
  title = {{MNIST} handwritten digit database},
  url = {http://yann.lecun.com/exdb/mnist/},
  username = {mhwombat},
  year = 2010
}

@article{Garofolo1993DarpaTA,
  title={DARPA TIMIT acoustic-phonetic continuous speech corpus},
  author={Lyons, John W},
  journal={National Institute of Standards and Technology},
  year={1993}
}

@article{Glazer2024FrontierMathAB,
  title={FrontierMath: A Benchmark for Evaluating Advanced Mathematical Reasoning in AI},
  author={Elliot Glazer and Ege Erdil and Tamay Besiroglu and Diego Chicharro and Evan Chen and Alex Gunning and Caroline Falkman Olsson and Jean-Stanislas Denain and Anson Ho and Emily de Oliveira Santos and Olli J{\"a}rviniemi and Matthew Barnett and Robert Sandler and Matej Vrzala and Jaime Sevilla and Qiuyu Ren and Elizabeth Pratt and Lionel Levine and Grant Barkley and Natalie Stewart and Bogdan Grechuk and Tetiana Grechuk and Shreepranav Varma Enugandla and Mark Wildon},
  journal={ArXiv},
  year={2024},
  volume={abs/2411.04872},
  url={https://api.semanticscholar.org/CorpusID:273877467}
}

@article{Huan2025DoesMR,
  title={Does Math Reasoning Improve General LLM Capabilities? Understanding Transferability of LLM Reasoning},
  author={Maggie Huan and Yuetai Li and Tuney Zheng and Xiaoyu Xu and Seungone Kim and Minxin Du and Radha Poovendran and Graham Neubig and Xiang Yue},
  journal={ArXiv},
  year={2025},
  volume={abs/2507.00432},
  url={https://api.semanticscholar.org/CorpusID:280146966}
}

@article{arc,
  title={Think you have Solved Question Answering? Try ARC, the AI2 Reasoning Challenge},
  author={Peter Clark and Isaac Cowhey and Oren Etzioni and Tushar Khot and Ashish Sabharwal and Carissa Schoenick and Oyvind Tafjord},
  journal={ArXiv},
  year={2018},
  volume={abs/1803.05457},
  url={https://api.semanticscholar.org/CorpusID:3922816}
}

@article{commonsenseqa,
  title={CommonsenseQA: A Question Answering Challenge Targeting Commonsense Knowledge},
  author={Alon Talmor and Jonathan Herzig and Nicholas Lourie and Jonathan Berant},
  journal={ArXiv},
  year={2019},
  volume={abs/1811.00937},
  url={https://api.semanticscholar.org/CorpusID:53296520}
}

@article{gsm8k,
  title={Training Verifiers to Solve Math Word Problems},
  author={Karl Cobbe and Vineet Kosaraju and Mo Bavarian and Mark Chen and Heewoo Jun and Lukasz Kaiser and Matthias Plappert and Jerry Tworek and Jacob Hilton and Reiichiro Nakano and Christopher Hesse and John Schulman},
  journal={ArXiv},
  year={2021},
  volume={abs/2110.14168},
  url={https://api.semanticscholar.org/CorpusID:239998651}
}

@inproceedings{hellaswag,
  title={HellaSwag: Can a Machine Really Finish Your Sentence?},
  author={Rowan Zellers and Ari Holtzman and Yonatan Bisk and Ali Farhadi and Yejin Choi},
  booktitle={Annual Meeting of the Association for Computational Linguistics},
  year={2019},
  url={https://api.semanticscholar.org/CorpusID:159041722}
}

@inproceedings{medmcqa,
  title={MedMCQA : A Large-scale Multi-Subject Multi-Choice Dataset for Medical domain Question Answering},
  author={Ankit Pal and Logesh Kumar Umapathi and Malaikannan Sankarasubbu},
  booktitle={ACM Conference on Health, Inference, and Learning},
  year={2022},
  url={https://api.semanticscholar.org/CorpusID:247763070}
}

@article{nq_open,
  title={Natural Questions: A Benchmark for Question Answering Research},
  author={Tom Kwiatkowski and Jennimaria Palomaki and Olivia Redfield and Michael Collins and Ankur P. Parikh and Chris Alberti and Danielle Epstein and Illia Polosukhin and Jacob Devlin and Kenton Lee and Kristina Toutanova and Llion Jones and Matthew Kelcey and Ming-Wei Chang and Andrew M. Dai and Jakob Uszkoreit and Quoc V. Le and Slav Petrov},
  journal={Transactions of the Association for Computational Linguistics},
  year={2019},
  volume={7},
  pages={453-466},
  url={https://api.semanticscholar.org/CorpusID:86611921}
}

@inproceedings{piqa,
  title={PIQA: Reasoning about Physical Commonsense in Natural Language},
  author={Yonatan Bisk and Rowan Zellers and Ronan Le Bras and Jianfeng Gao and Yejin Choi},
  booktitle={AAAI Conference on Artificial Intelligence},
  year={2019},
  url={https://api.semanticscholar.org/CorpusID:208290939}
}

@misc{social_iqa,
      title={SocialIQA: Commonsense Reasoning about Social Interactions}, 
      author={Maarten Sap and Hannah Rashkin and Derek Chen and Ronan LeBras and Yejin Choi},
      year={2019},
      eprint={1904.09728},
      archivePrefix={arXiv},
      primaryClass={cs.CL},
      url={https://arxiv.org/abs/1904.09728}, 
}

@misc{winogrande,
      title={WinoGrande: An Adversarial Winograd Schema Challenge at Scale}, 
      author={Keisuke Sakaguchi and Ronan Le Bras and Chandra Bhagavatula and Yejin Choi},
      year={2019},
      eprint={1907.10641},
      archivePrefix={arXiv},
      primaryClass={cs.CL},
      url={https://arxiv.org/abs/1907.10641}, 
}

@article{tal2020multiagent, 
title={Multiagent Evaluation Mechanisms}, volume={34}, url={https://ojs.aaai.org/index.php/AAAI/article/view/5543}, 
number={02}, 
journal={Proc.\ {AAAI} Conference on Artificial Intelligence}, 
author={Alon, Tal and Dobson, Magdalen and Procaccia, Ariel and Talgam-Cohen, Inbal and Tucker-Foltz, Jamie}, 
year={2020}, month={Apr.}, pages={1774-1781} 
}

@article{connelly2014tournament,
author = {Brian L. Connelly and Laszlo Tihanyi and T. Russell Crook and K. Ashley Gangloff},
title ={Tournament Theory: Thirty Years of Contests and Competitions},
journal = {Journal of Management},
volume = {40},
number = {1},
pages = {16-47},
year = {2014},
doi = {10.1177/0149206313498902},
}

@misc{fu2019contests,
      author = "Qiang Fu and Zenan Wu",
      title = "Contests: Theory and Topics",
      year = "2019",
      month = "07",
      publisher = "Oxford University Press",
      doi = "10.1093/acrefore/9780190625979.013.440",
      url = "https://oxfordre.com/economics/view/10.1093/acrefore/9780190625979.001.0001/acrefore-9780190625979-e-440"
}

@article{corchon2007theory,
	author = {Corch{\'o}n, Luis C. },
	id = {Corch{\'o}n2007},
	journal = {Review of Economic Design},
	number = {2},
	pages = {69--100},
	title = {The theory of contests: a survey},
	url = {https://doi.org/10.1007/s10058-007-0032-5},
	volume = {11},
	year = {2007},
}

@misc{liberman2015reproducible,
  author = {Liberman, Marc},
  journal = {Simons Foundation Lecture},
  title = {Reproducible research and the common task method},
  url = {https://www.simonsfoundation.org/lecture/reproducible-research-and-thecommon-task-method},
  year = {2015}}
